%% file: report.tex
\documentclass[12pt]{article}
\pdfoutput=1
\usepackage[margin=0.8in]{geometry}
\usepackage{amsmath,amssymb,latexsym}
\usepackage{graphicx} 
\usepackage{booktabs}
\usepackage{hyperref}
\usepackage[numbers,sort&compress]{natbib}
\usepackage[toc,page]{appendix}

\usepackage{amsthm}
\newtheorem{theorem}{Theorem}
\usepackage{authblk}
\DeclareMathOperator*{\argmin}{argmin}
\newtheorem{lemma}[theorem]{Lemma}
\usepackage{wrapfig}

\input{abbr}

\graphicspath{{figs/}}

\title{Simultaneous Learning of the Inputs and  Parameters in\\ Neural Collaborative Filtering}

\author[ \hspace{-1ex}]{Ramin Raziperchikolaei}
\author[ \hspace{-1ex}]{Young-joo Chung}
\affil[ ]{Rakuten Group, Inc.}
\affil[ ]{\textit {\{ramin.raziperchikola,youngjoo.chung\}@rakuten.com}}
\date{} 

\hypersetup{
pdftitle={Simultaneous Learning of the Inputs and  Parameters in\\ Neural Collaborative Filtering},
pdfauthor={Ramin Raziperchikolaei, Young-joo Chung},
pdfkeywords={Simultaneous Learning of the Inputs and  Parameters in\\ Neural Collaborative Filtering}
}

\begin{document}

\maketitle
\begin{abstract}
Neural network-based collaborative filtering systems focus on designing network architectures to learn better representations while fixing the input to the user/item interaction vectors and/or ID. In this paper, we first show that the non-zero elements of the inputs are learnable parameters that determine the weights in combining the user/item embeddings and fixing them limits the power of the models in learning the representations. Then, we propose to learn the value of the non-zero elements of the inputs jointly with the neural network parameters. We analyze the model complexity and the empirical risk of our approach and prove that learning the input leads to a better generalization bound. Our experiments on several real-world datasets show that our method outperforms the state-of-the-art methods, even using shallow network structures with a smaller number of layers and parameters.
\end{abstract}


\section{Introduction}
\label{s:intro}
With the explosive growth of online information in the past few decades, recommender systems (RSs) have become imperative to overcome the information overload problem \citep{Ricci11}. The goal of RSs is to predict user feedback on items, which can be helpful in providing a list of suggested items to a user. The feedback (also called interaction) can be implicit (e.g., purchased) or explicit (e.g., a rating between $1$ to $5$). Our approach can be applied to both types of feedback. The "interaction" throughout this paper refers to both types.

Collaborative filtering (CF), one of the most widely used and studied approaches in RSs, learns users' preferences from the patterns in the past feedback between the users and items \citep{Ricci11}. Matrix factorization (MF), the most popular approach in CF, models the interaction between the users and items by the dot product of their representations \citep{Koren09}. 

Neural network-based collaborative filtering methods have recently become popular to extend the idea of the MF and achieve better predictions. These methods involve neural networks in learning the user/item representations and/or modeling the interactions \citep{He18}. The main neural network architectures used in these models are autoencoders \citep{Sedhain15}, multilayer perceptrons (MLPs) \citep{Xue17,He17,Dong19}, and convolutional neural networks (CNNs) \citep{He18}.

In this paper, we focus on another key--but lesser studied--element of neural network-based methods: \emph{the input layer and its role in learning better representations.} Two types of inputs are used in neural network-based RSs. The first one is the user/item ID, where each ID is converted to a one-hot-encoding vector and gets connected to an embedding matrix \citep{He17,He18}. The second one is the user-item interaction vector, where the user interaction vector contains user feedback on all items and the item interaction vector contains all feedback given to an item \citep{Sedhain15,Xue17,Dong19}. Even the hybrid RS models, which utilize side information to overcome the issue of interaction sparsity, still keep the user/item ID and/or interaction vector to get competitive results \citep{Strub16,Zhang17,Dong17}. 

Table 2 of \citet{Dong19} shows that DeepCF (with the interaction vector as the input) outperforms NeuMF (with the ID as the input) in four different datasets. Table 2 of  \citet{Xue17} and our experiments confirm these results. In Section~\ref{s:ncf}, we give the details of the NCF with the interaction vector and ID as the input and explain why one works better than the other.

It is known that deep neural networks can be applied to images and textual data, which are spatially and/or temporally correlated, to achieve state-of-the-art results. As pointed out by \citet{Lian18}, the input data in RSs is different from images and textual data.  First, there is no temporal or spatial correlation in the interaction vectors. Second, the interaction vectors are massively sparse (between $95\%$ to $99\%$) and have vast dimensions (between a few thousand to a few million). What kind of information does the neural network extract from the interaction vectors to learn user and item representations then?

To answer this question, we utilize Neural Collaborative Filtering (NCF) framework proposed by \citet{He17} to understand and analyze the role of interaction vectors. We selected NFC because it is a basic two-branch architecture for neural collaborative filtering. In most NCF papers, the input layer is connected to a set of fully connected layers \citep{Xue17,He17, Dong19}. We show that the weight matrix of the first layer can be considered as an implicit user/item embedding matrix, and the non-zero elements of the interaction vector determine which embedding vectors to choose.  We argue that the non-zero elements of the interaction vectors are \emph{learnable parameters} that determine the weights in combining the embeddings. Fixing these input parameters limits the power of the models in learning the representations. To achieve better representations, we propose to learn the input and neural network parameters jointly. We theoretically analyze our approach and prove that it achieves a better empirical error and a lower generalization bound than  previous works with the fixed input.

We conduct extensive experiments on several real-world datasets. The experiments show that our approach has a better prediction performance than state-of-the-art methods, even using shallow network structures with a smaller number of parameters.

\section{Related work}
\label{s:related}
Matrix factorization (MF) is the most popular collaborative filtering approach. MF uses the dot product of the user and item representations to estimate their interaction \cite{Koren08,Takacs08,Koren09}. To improve performance, neural networks have been involved in learning representations and modeling interactions \cite{Cheng16,  He17a, Guo17,He17,Dong19,Xue17,He18}. NeuMF \cite{He17} takes the user and item IDs as the inputs and uses embedding layers to create representations. Two joint user-item representations are then learned, one by the element-wise product of the user and item representations and the other one by applying an MLP to their concatenation. Finally, another MLP takes the two joint representations and predicts the final rating. DeepCF \cite{Dong19} makes two changes to the NeuMF to learn the representations. First, DeepCF uses user/item interaction vectors as the input, instead of IDs. Second, before applying the element-wise product, it uses several MLPs on top of the user and item representations. DMF \cite{Xue17} uses explicit ratings as the input to a set of MLPs to learn the representations. The user-item interactions are modeled with the dot product of the representations. ConvNCF \cite{He18} learns user and item representations by connecting user/item IDs to the embedding layers and applies convolutional neural networks to their outer product to predict the ratings. Autorec is an autoencoder-based CF method that reconstructs the known values of the interaction vectors to predict unknown ones \cite{Sedhain15}. 

Autoencoders have been used extensively in hybrid collaborative filtering to utilize various sources of side information for the users and items, such as users' age, users' occupation, and items' title.\ \cite{Li15,Wang15b,Strub16,Dong17,Zhang17b,Li18,Steck20}. These methods try to reconstruct the input vectors (fixed to the interaction vector and side information) to learn representations for the prediction task.

The focus of these previous works is on designing complicated network structures to learn high-quality representations. A recent work \cite{Chen20} uses a different approach and focuses on the quality of the output ratings in collaborative filtering. It simultaneously learns the model parameters and edits the overly personalized ratings using data debugging techniques. 

All the previous works fix the input to either the interaction vector or the IDs.  On the contrary, we focus on the input layer and its role in learning better representations.

\paragraph{Notations.}
We denote the sparse interaction matrix by $\bR \in \mathbb{R}^{m \times n}$, where $m$ and $n$ are the number of users and items, respectively, $R_{jk}>0$ is the interaction value of the user $j$ on the item $k$, and $R_{jk}=0$ means the interaction is unknown.  The goal is to predict the unknown interactions in $\bR$. The $i$th row of a matrix $\bH$ is shown by $\bH_{i,:}$ and the $j$th column is shown by $\bH_{:,j}$.

\section{Neural collaborative filtering}
\label{s:ncf}
Since our proposed method uses the neural collaborative filtering (NCF) framework, we review NCF in this section. In NCF framework, as depicted in Fig.~\ref{f:inputs}(a), the interaction of the users on items is predicted by an MLP, which takes the user and item representations as the input. 

Let us denote the $j$th user representation by $\smash{\bz^{u}_j \in \bbR^{d_u}}$ and the $k$th item representation by $\smash{\bz^{i}_k \in \bbR^{d_i}}$. These two representations are concatenated to give the user-item representation, which is denoted by $\smash{\bz_{jk}=[\bz^{u}_j,\bz^{i}_k] \in \bbR^{d_u+d_i}}$. Then, an MLP, denoted by  $f()$ , takes $\bz_{jk}$  and predicts the interaction. The objective function is defined as:
\begin{equation}
	\label{e:ncf}
E_{\text{NCF}}(\bz_{jk},\btheta^{ui})= \frac{1}{mn} \sum_{j=1}^{m}\sum_{k=1}^n  L(R_{jk},f(\bz_{jk};\btheta^{ui}))
\end{equation}
where $\btheta^{ui}$ contains all the weight matrices of the MLP $f()$. Minimization of the loss function $L()$ should lead to the match between the observed and predicted interactions. Both regression \cite{He17} and classification (binary cross-entropy) \cite{Dong19} losses have been used in the literature, which are defined as follows:
\begin{equation}
	L_{\text{reg}} =  ||R_{jk} - \hat{R}_{jk}||^2, \qquad \qquad L_{\text{bce}} =  - (R_{jk} \log\hat{R}_{jk} + (1-R_{jk}) \log(1-\hat{R}_{jk}) ),
\end{equation}
where  $\hat{R}_{jk} = f(\bz_{jk};\btheta^{ui})$ is the predicted interaction.

The focus of our paper is on learning better user and item representations, which play an important role in the final performance of the recommendation systems. In the rest of this section, we explain how these representations are learned with the ID and interaction vector as the input.

\begin{figure}[!t]
	\begin{center}
	\begin{tabular}{c@{\hspace{5ex}}c@{\hspace{5ex}}c}	\includegraphics[height=0.35\linewidth]{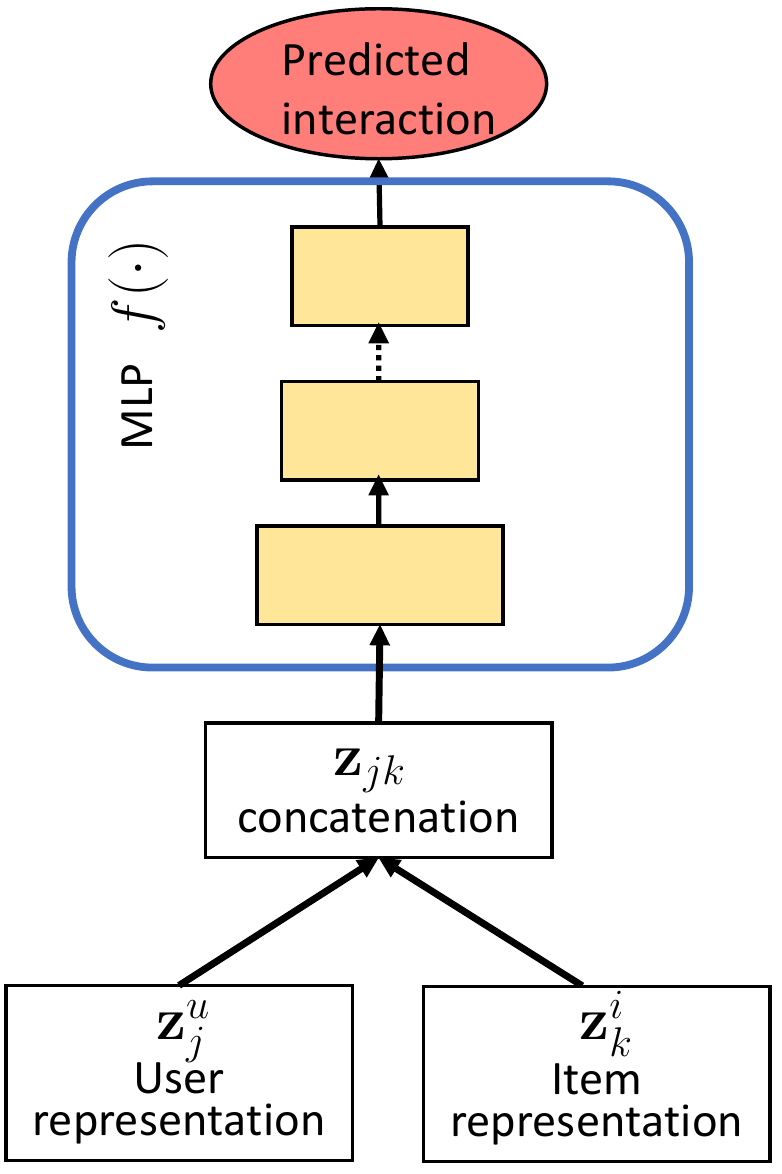} & 						     \includegraphics[height=0.35\linewidth]{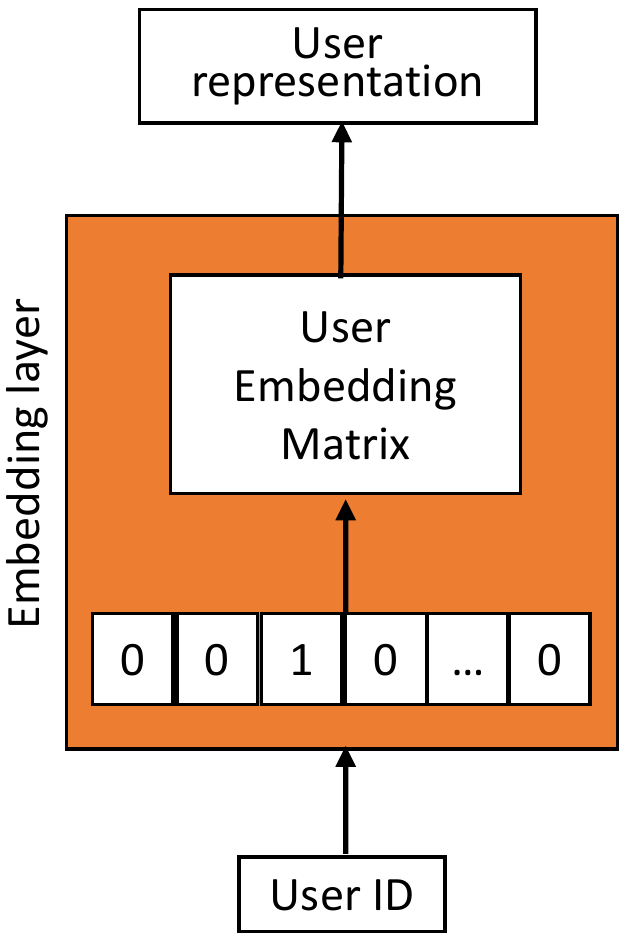} &    \includegraphics[height=0.35\linewidth]{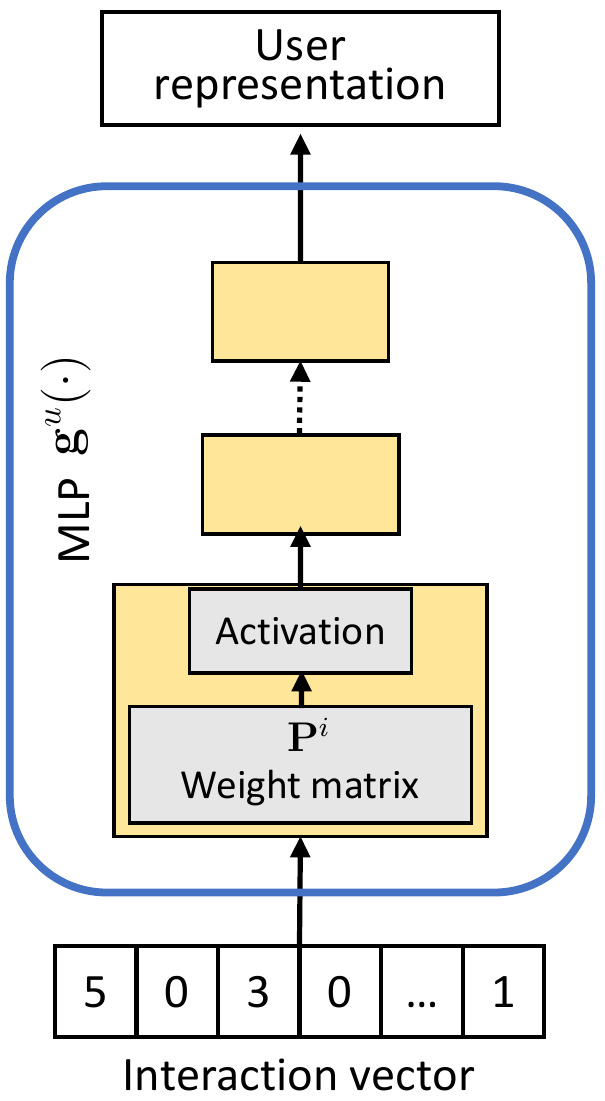} \\
			(a) & (b) & (c) \\
			\end{tabular}
			\caption{ (a) NCF applies MLP to the concatenation of the user and item representations to predict the interaction. (b) Representation learning with the ID as the input. (c) Representation learning with the interaction vector as the input.}
			\label{f:inputs}
	\end{center}
\end{figure}

\subsection{NCF with the ID as the input}
Figure~\ref{f:inputs}(b) shows how IDs are mapped to the low-dimensional representations. First, the user ID is converted to a one-hot-encoding vector of size $m$ (number of users). For user $j$, this vector is denoted by $\bI^{u}_j \in \bbR^{m}$, where the $j$th entry of $\bI^{u}_j$ is $1$ and the rest are $0$. Then, this vector is multiplied by a user embedding matrix $\smash{\bE^{u} \in \bbR^{m\times d_u}}$ to give us the $d_u$ dimensional representation vector, i.e., $\smash{\bz^{u}_j = \bI^{u}_j\bE^{u}}$. We can write the $k$th item representation in the same way as $\smash{\bz^{i}_k = \bI^{i}_k\bE^{i}}$, where $\bI^{i}_k \in \bbR^{1\times n}$, and $\smash{\bE^{i} \in \bbR^{n\times d_i}}$ is the item embedding matrix. The objective function of the NCF with the ID as the input is defined as:
\begin{equation}
	\label{e:id}
	\begin{split}
&E_{\text{NCF}}(\bE^{u},\bE^{i},\btheta^{ui})= E_{\text{NCF}}(\bz_{jk},\btheta^{ui})= \frac{1}{mn} \sum_{j=1}^{m}\sum_{k=1}^n  L(R_{jk},f(\bz_{jk};\btheta^{ui})) \\
	&\text{s.t.} \quad  \bz_{jk} = [\bz^{u}_j,\bz^{i}_k], \quad \bz^{u}_j = \bI^{u}_j\bE^{u}, \quad \bz^{i}_k = \bI^{i}_k\bE^{i}.
	 \end{split}
\end{equation}
\subsection{NCF with the interaction vector as the input}
Fig.~\ref{f:inputs}(c) shows how the interaction vector as the input works. First, for the $j$th user, the $j$th row of the interaction matrix $\bR$ is extracted. An MLP, denoted by $\bg^{u}()$, takes this vector and generates the user representation:
\begin{equation*}
\bz^{u}_j = \bg^{u}(\bR_{j,:};\btheta^{u}_{L}) = \sigma(\dots\sigma(\sigma(\bR_{j,:}\bP^{u})\bW^{u}_2)\dots \bW^{u}_L),
\end{equation*}
where $L$ is the number of layers, $\sigma()$ is an activation function, $\smash{\btheta^{u}_{L}=\{\bP^{u},\bW^{u}_2,\dots,\bW^{u}_L \}}$ contains the weights of all layers, and $\bP^{u}$ and $\bW^{u}_l$ are the weight matrices of the first and $l$th layers ($l>1$), respectively. We can write the $k$th item representation in the same way as $\smash{\bz^{i}_k = \bg^{i}(\bR_{:,k};\btheta^{i}_{L})}$.  Note that for simplicity, we assume both $\bg^u()$ and $\bg^i()$ have $L$ layers and we ignore the biases. We can write the objective function as:
\begin{equation}
	\label{e:R}
	\begin{split}
&E_{\text{NCF}}(\btheta^{u}_L,\btheta^{i}_L,\btheta^{ui})= E_{\text{NCF}}(\bz_{jk},\btheta^{ui}) = \frac{1}{mn} \sum_{j=1}^{m}\sum_{k=1}^n  L(R_{jk},f(\bz_{jk};\btheta^{ui})) \\
	&\text{s.t.} \quad  \bz_{jk} = [\bz^{u}_j,\bz^{i}_k], \quad \bz^{u}_j = \bg^{u}(\bR_{j,:};\btheta^{u}_L), \quad \bz^{i}_k = \bg^{i}(\bR_{:,k};\btheta^{i}_L)
	 \end{split}
\end{equation}
where $\btheta=[\btheta^{u}_L,\btheta^{i}_L,\btheta^{ui}]$ contains all parameters of the $\bg^u(),\bg^i()$, and $f()$ MLPs, respectively. 

\subsection{ID versus interaction vector as the input}
Table 2 of \citet{Dong19} and Table 2 of \citet{Xue17} show that DeepCF (interaction vector as the input) outperforms NeuMF (ID as the input) in four different datasets. Our experiments also confirm these results. We also found that DeepCF converges faster than NeuMF. To understand the reason of performance difference between ID and interaction vectors, we examine the toy interaction matrix of Fig.~\ref{f:toy}. Here, there are three users and four items, $1$ means implicit feedback and $0$ means unknown feedback.

The main idea of collaborative filtering is to use the patterns in the past interactions of the users and items to predict future interactions. In Fig.~\ref{f:toy}, we can see that the past interactions of user\#1 and user\#2 are similar, as they have both interacted with item\#1 and item\#2. However, there is no common item in the past interactions of the user\#2 and user\#3. In CF, we expect user\#1 and user\#2 to follow the same pattern, but not much can be said about the similarity of user\#2 and user\#3. Let us see how this information appears in the input spaces created by the IDs and interaction vectors.

In the case of ID as the input, each user ID is converted to a one-hot-encoding of length four. So the Euclidean distance between each pair of users is $\sqrt{2}$. In other words, all users are at the same distance from each other. 

In the case of the interaction vector as the input, each row of the interaction matrix of Fig.~\ref{f:toy} corresponds to the input vector of each user. We can see that the Euclidean distance between the input vectors of the users \#1 and \#2 (similar past interactions) is $1$, while the distance between the input vectors of users \#2 and \#3 (dissimilar or unrelated past interactions) is $\sqrt{3}$. This example shows that the similarity is preserved (to some extent) in the input space created by the interaction vectors, which is the reason that NCF with interaction vector as input outperforms NCF with ID as input.

In the rest of this script, we focus on NCF with the interaction vectors as the input and how we can improve its performance.

\begin{figure}
	\begin{center}
		\begin{tabular}{ccccc}
		\toprule
		& item\#1 & item\#2  & item\#3 & item\#4 \\
		\midrule
		user\#1 &1 & 1 & 1 & 0\\
		user\#2 & 1 & 1 & 0 & 0\\
		user\#3 & 0 & 0 & 0 & 1
		\end{tabular}	
		\caption{Toy example.  Interaction matrix of three users and four items, where $1$ is observed implicit feedback and $0$ is unknown feedback.}
		\label{f:toy}
		\end{center}
\end{figure}
\section{Proposed method: learning the inputs in neural\\ collaborative filtering}
\label{s:prop}

We explain our proposed approach using the NCF architecture introduced in Section \ref{s:ncf}. This makes our equations simpler and enable us to focus on our main idea of learning the inputs better. In the supplementary material in Section~\ref{s:sup}, we change the architecture to the more complicated NCF architectures of CFNet \cite{Dong19} and show that our arguments and theoretical analysis remain valid. Our experiments confirm that our approach improves the results of the previous works, regardless of the specific type of NCF architecture.

We motivate the idea of learning the inputs by explaining what the MLPs in Eq.~\eqref{e:R} learn from the interaction vectors. While there are several nonlinear layers in $\bg^{u}()$ and $\bg^{i}()$, we focus on understanding the output of the first layer, which gives us a good intuition about how the interaction vectors are used. 

Consider the first fully connected layer of the $\bg^{u}()$, with the weight matrix $\bP^{u}$. Assuming linear activation function, the output of this layer can be written as $\by^{u}_j=\bR_{j,:}\bP^{u}=\sum_{l=1}^{n}\bR_{j,l}\bP^{u}_{l,:} $. Note that $\bP^{u} \in \bbR^{n\times d_u}$ contains a $d_u$ dimensional vector for each of the items. In other words, we can consider $\bP^{u}$ as an (implicit) item embedding matrix. \emph{The user (initial) representation $\by^{u}_j$ is the weighted sum of the embedding vectors of the items which are interacted by the $j$th user.} Similarly for the items, the weights of the first layer of $\bg^{i}()$, denoted by $\bP^{i} \in \bbR^{m \times d_i}$, can be considered as an implicit user embedding layer, and \emph{the item  (initial) representation $\by^{i}_k$ is the weighted sum of the embedding vectors of the users who interacted with the $k$th item.} In previous works, the weights in the "weighted sum of embedding vectors" are fixed to the interaction values (ratings or ones) of the users on the items. We will show that fixing these weights limits the power of the representations and we should consider them as the learnable parameters.

Previous works rely on the nonlinear activation functions, making networks deep, and creating multiple user and item representations (from different network architectures) to make powerful representations.  On the contrary, our approach learns the \textbf{In}put and \textbf{P}arameters of \textbf{NCF} simultaneously, and is denoted by \textbf{InP-NCF} in the rest of the paper. 

Next, we define our objective function formally. For the $j$th user, we define the set $N_j=\{l|R_{jl} \ne 0, 1 \le l \le n\}$, which contains the index of all the items interacted by the $j$th user. We define the set of user input vectors as $\bU \in \bbR^{m\times n}$, where $\bU_{j,:}$ is the input vector of the $j$th user. We learn $U_{jl}$ if $l \in N_j$ and fix it to $0$ otherwise. The output of the first layer of the network for the $j$th user is computed as $\by_{j}=\sigma(\bU_{j,:}\bP^{u}) = \sigma(\sum_{l \in N_j}  U_{jl} \bP^{u}_{l,:})$, where we consider $U_{jl}$ as a learnable parameter, as opposed to the previous works that fix it to the interaction value of the $j$th user on the $l$th item. For the $k$th item, we define the set $N_k=\{l|R_{lk} \ne 0, 1 \le l \le m\}$, which contains the index of all the users who interacted with the $k$th item. We also define the set of item input vectors as $\bV \in \bbR^{n \times m}$. We learn $V_{kl}$ if $l \in N_k$ and fix it to $0$ otherwise. Our objective function is defined as follows:
\begin{equation}
	\label{e:ours}
\hspace{-.5ex}	E_{\text{InP-NCF}}(\btheta,\bU,\bV) =  E_{\text{NCF}}(\btheta^{u}_L,\btheta^{i}_L,\btheta^{ui})\ \text{s.t.}\  \begin{cases} 	
	\bz^{u}_j = \bg^{u}(\by^{u}_j;\btheta^{u}_{L-1}), \  \by^{u}_{j}=\sigma(\sum_{l \in N_j}  U_{jl} \bP^{u}_{l,:}) \\
	\bz^{i}_k = \bg^{i}(\by^{i}_k;\btheta^{i}_{L-1}), \ \by^{i}_{k}=\sigma(\sum_{l \in N_k} V_{kl} \bP^{i}_{l,:})
\end{cases}\hspace{-2ex}
\end{equation}
where $\btheta=[\btheta^{u}_{L-1},\bP^{u},\btheta^{i}_{L-1},\bP^{i},\btheta^{ui}]$ contains all the parameters of the three MLPs, $\bP^{u}$ and $\bP^{i}$ contain the weights of the first layer of the user and item MLPs, respectively, and $\by^{u}_j$ and $\by^{i}_k$ are the output of the first layer of the user and item MLPs, respectively. By setting $\smash{U_{jl} = R_{jl}}$ and $\smash{V_{kl}=R_{lk}}$, our objective in Eq.~\eqref{e:ours} becomes equivalent to the previous works' objective function in Eq.~\eqref{e:R}.

While both $\bU$ and $\bV$ contain $m\times n$ entries, most entries are $0$ because of the sparsity of the interaction matrix and don't get updated during training. We learn $U_{jk}$ and $V_{kj}$ if user $j$ interacted with item $k$. So the total number of learnable variables in $\bU$ and $\bV$ is twice the number of observed interactions, which is insignificant compared to the number of network parameters in large datasets. We give more details about the number of model and input  parameters in the experimental results in Section.~\ref{s:exp} and in Table~\ref{t:numP}.

We apply end-to-end optimization to learn the input and model parameters jointly. We found that the model parameters overfit faster than the input parameters. As we will show in our experiments, we can improve results by post-input optimization. In other words, at the end of the joint optimization, we fix the network parameters and continue learning the input parameters.

\subsection{Our approach from the lens of the graph neural networks}
We first show that neural collaborative filtering can be formulated as a graph convolutional network with a one-layer propagation and a deep message aggregation function. Then we show how our input-learning approach works and what learning the inputs means in this context.

Matrix completion can be cast as a link prediction problem on a bipartite user-item graph \cite{Berg18,Wang19a,Sun19,Ying18,He20}. This graph contains $m$ user nodes and $n$ item nodes. An edge between a user node $j$ and item node $k$ exists if  $R_{jk} \ne 0$. The graph-based methods extract node features using methods such as graph convolutional neural networks (GCNs) and then map them to the interaction value.

NCF has all the three main components of the GCNs-based recommendation systems methods:
\begin{enumerate}
\item  \textbf{Node embeddings.} In NCF, each row of the matrices $\bP^u$ and $\bP^i$ contains the embeddings of the user and item nodes, respectively.  
\item \textbf{Propagation layer.} This layer defines how the messages between the nodes are constructed and how they get aggregated. In NCF, for a connected user-item pair $(j,k)$, the message from user $j$ to item $k$ is the user's embedding $\bP^u_{j,:}$, and the message from item $k$ to user $j$ is the item's embedding $\bP^i_{k,:}$. To understand the message aggregation process in NCF, let us divide it into two steps. In the first step, the weighted sum  of the messages is computed and followed by an activation function:
\begin{equation}
\label{e:prop}
	 \by^{u}_{j}=\sigma(\sum_{l \in N_j}  R_{jl} \bP^{i}_{l,:}) \quad  \by^{i}_{k}=\sigma(\sum_{l \in N_k} R_{lk} \bP^{u}_{l,:}).
\end{equation}
As the second step, NCF applies the last $L-1$ layers of the $\bg^u()$ and $\bg^i()$ to the  $\by^{u}_{j}$ and $\by^{i}_{k}$, respectively, to get more complicated representations. The GCNs-based methods use different variants of the  Eq.~\eqref{e:prop} as the aggregation function \citep{Berg18,Wang19a,Sun19,Ying18,He20}.

\item \textbf{Prediction model.} NCF uses an MLP $f()$   to map the user and item representation to the final interaction value. Dot-product of the user and item representation is one of the popular prediction models in previous GCNs-based works \cite{Wang19a,Sun19,He20}.
\end{enumerate}

The main difference between NCF and GCNs is how they use neural networks to learn better representations: GCNs use higher-order propagations, while NCF relies on the power of the MLPs $\bg^u()$ and $\bg^i()$.

Let us explain how our approach works in this context. InP-NCF improves the first step of the message aggregation of NCF: instead of fixing the weights of each message to the interaction values $R_{jl}$ and $R_{lk}$, InP-NCF learns these weights, which leads to learning better representations.

Using the explicit ratings in learning better representations has been tried in graph-based recommendation systems. \citet{Ying18} learns a specific transformation (weight matrix) for each rating level, resulting in edge-type specific messages. \citet{Fan19} learns a specific embedding vector for each rating level and uses them in learning the item representations. In both methods, similar ratings get similar embeddings or transformation functions, even if the ratings belong to different users or items. InP-NCF considers each rating as a learning parameter and let the similar ratings of different users/items converge to different places. This makes the model learned by InP-NCF more powerful.

\subsection{Attention mechanisms in learning the weights}
The goal here is to explain the difference between our approach InP-NCF and the attention-based methods GAT \cite{Velickovic18} and NAIS \cite{He18a}.

GAT \cite{Velickovic18} showed that assigning different importances to nodes of the same neighborhood can lead to better results in GCNs. GAT uses self-attention to find the weights in computing the weighted sum in aggregating the neighbor messages. In recommendation system literature, NAIS \cite{He18a}, uses attention mechanisms to give different weights to the item embeddings in computing the users' initial representations.

Let us explain the advantages of the InP-NCF. In many recommendation system tasks, we have access to the explicit interactions, which already determine the importance of the neighbors to some extent. In other words, if $R_{jk}>R_{jl}$, then item $k$ is more favorable for user $j$.  Our InP-NCF approach uses this information by initializing the weights with the user and item interactions. GAT and NAIS do not use this information in defining the weights. Also, since the weights are free parameters in InP-NCF and we can initialize them form the given interaction values, we can prove that InP-NCF achieves a smaller empirical error and a better generalization bound than NCF (next subsection). Since weights in attention-based methods are a function of the embeddings, we can not initialize them directly with the ratings.

Here is a detailed explanation of the differences between NAIS and our approach:
\begin{itemize}
\item NAIS relies on item IDs as the input and focuses on the item-item similarity CF models to predict the users' interests. Our work uses item and user interaction vectors as the inputs, learns separate user and item representations, and uses user-item CF to predict the score. In other words, the two methods use completely different structures and follow different CF models.

\item Since NAIS focuses on the item-item similarity, the weights are only defined for the "items" which are rated by the users. There is no easy way to define the weights for the "users" who rated the items. Our approach defines weights for users and items separately. 

\item The items' weights in NAIS are defined based on the similarity to the \emph{target} item. In our case, both user and item weights are \emph{free variables.}
\item Finally, as we will show in experiments of Table~\ref{t:comp_imp}, InP-NCF achieves better results than NAIS.
\end{itemize}

\subsection{Generalization bound of InP-NCF versus NCF}
\label{s:gb}
In this section, we first define the basic elements of the statistical learning theory in collaborative filtering. Then, we show that we can compare generalization bound of the INP-NCF and NCF based on their empirical errors. More details about these definitions can be found in \cite{Shwartz14}.

\paragraph{The collaborative filtering model and the hypothesis class.}
In collaborative filtering (matrix completion), the ultimate goal is to predict the interaction value for a pair of users and items. The assumption is that we have has seen all users and items and a small subset of their interactions in the training set. In other words, there would be no new user or item in the test set. The prediction model of the latent factor-based collaborative filtering methods can be written as: 
\begin{equation}
	\label{e:zmodel}
	\hat{R}_{jk} =  f(\bz_{jk};\btheta^{ui}) \quad \text{s.t.} \quad
	\bz_{jk}=\bm(\bz^u_j,\bz^i_k), \   \bz^{u}_j = \bI^{u}_j\bZ^{u}, \ \bz^{i}_k = \bI^{i}_k\bZ^{i},
\end{equation}
where $ \bI^{u}_j$ and $\bI^{i}_k$ are the one-hot encoding vectors of the $j$th user and $k$th item and $\bm()$ is a fusion function that takes two representations and returns a joint one. The input to this model is the user and item IDs and the parameters are $\btheta^{ui}$, $\bZ^u$, and $\bZ^i$. In statistical learning theory, the model of Eq.~\eqref{e:zmodel} is called the hypothesis class. Different choices for the $f()$ and the fusion function $\bm()$ lead to different hypothesis classes. In matrix factorization, $\bm()$ is an element-wise product and $f()$ is a summation. In NCF, $\bm()$ is concatenation and $f()$ is an MLP. Each specific set of the parameters corresponds to a specific hypothesis.

\paragraph{The learning algorithm.}
The CF methods differ on the way they learn the parameters of the model in Eq.~\eqref{e:zmodel}. A learning algorithm, denoted by $\mathcal{A}$, takes the i.i.d.\ observations $\mathcal{D}$ and generates a specific set of parameters, which creates a hypothesis. In CF, the learning algorithm returns: $\btheta^{ui},\bZ^u,\bZ^i=\mathcal{A}(\mathcal{D})$. 

The latent factor-based CF methods in the literature, including NCF and INP-NCF, are learning algorithms. These methods use different training strategies, sets of parameters, models, and objective functions. At the end, all these methods return a hypothesis from the hypothesis class of Eq.~\eqref{e:zmodel}, i.e., they return the user and item representations and parameters of the $f()$.

While a learning algorithm must return parameters of a hypothesis, it is not limited to use only those parameters in training. InP-NCF and NCF (as learning algorithms) use the user and item MLPs $\bg^u()$ and $\bg^i()$ to learn better representations from interaction vectors or learnable inputs during the training. But, all these parameters are discarded afterward.

Let us mention that in addition to collaborative filtering, there exist other machine learning areas where the learning algorithm uses more parameters than the hypothesis space's parameters. One example is hyper-parameter optimization (HPO) \cite{Tobias16,Wei21}. In HPO, the learning algorithm learns two sets of parameters, one set to predict the target and another set to predict the best hyper-parameters (such as leanring rate and dropout ratio). But, the output of the learning algorithm is always the parameters of the hypothesis (prediction model).

\paragraph{The empirical risk.}
For each hypothesis, this risk is defined based on the difference between the predicted interactions and the true interactions on a training set (a set of observations from the true distribution). For a particular choice of the parameters $\btheta^{ui}$, $\bZ^u$, and $\bZ^i$, the empirical risk is defined as:
\begin{equation}
\label{e:emp_risk}
R_{emp}(\btheta^{ui},\bZ^u,\bZ^i) = \frac{1}{mn} \sum_{j=1}^{m}\sum_{k=1}^n  L(R_{jk},\hat{R}_{jk}),
\end{equation}
where $\hat{R}_{jk}$ is defined in Eq.~\eqref{e:zmodel}.

\paragraph{The generalization bound}
The expected  risk $E^{exp}$ (also called actual risk) is the expectation of the test error for a  hypothesis. Since we do not have access to the true data distribution, we can not compare the hypothesizes directly using $E^{exp}$. Therefore, the generalization bound is the upper bound on $E^{exp}$ and defined as follows \cite{Vapnik95}:
\begin{equation}
	\label{e:grbound}
	R^{emp}(\btheta^{ui},\bZ^u,\bZ^i) + C(H(\text{hypothesis space}), N),
\end{equation}
where $R^{emp}$ is the empirical risk (error on a set of observations), $C()$ is a cost function, $H()$ is complexity of the hypothesis space, and $N$ is the number of training samples. Here, we focus on the VC dimension-based bounds, where complexity is a function of the number of parameters of the model. The smaller the number of parameters the tighter this bound. We can make bound tighter by making the representations $\bZ^{u}$ and $\bZ^{i}$ low-dimensional and making $f()$ a shallow function.

\paragraph{Our main theorem}
\begin{theorem}
$\mathcal{A}_{\text{InP-NCF}}(\mathcal{D})$ achieves a generalization bound  smaller than or equal to the $\mathcal{A}_{\text{NCF}}(\mathcal{D})$.
\end{theorem}
\begin{proof}

Let us assume the learning algorithms return the following parameters at the convergence: 
\begin{equation}
{\btheta^{ui}}^*,{\bZ^{u}}^{*},{\bZ^{i}}^{*} = \mathcal{A}_{\text{InP-NCF}}(\mathcal{D}), \qquad \qquad  \hat{\btheta^{ui}},\hat{\bZ^{u}},\hat{\bZ^{i}} = \mathcal{A}_{\text{NCF}}(\mathcal{D}).
\end{equation}

Generalization bound is defined at the right side of Eq.~\eqref{e:grbound}. The theorem holds because: 1) $C()$ is determined based on the complexity of the hypothesis space of Eq.~\eqref{e:zmodel}, is independent of the learning algorithms, and is the same for both $\mathcal{A}_{\text{InP-NCF}}$ and $\mathcal{A}_{\text{NCF}}$, and 2) as we will prove in Lemma~\ref{l:emperr}, we have:
\begin{equation}
\label{e:rtempdiff}
R^{\text{emp}}_{\text{Inp-NCF}} = R^{\text{emp}}({\btheta^{ui}}^*,{\bZ^{u}}^{*},{\bZ^{i}}^{*}) \le   R^{\text{emp}}_{\text{NCF}}  = R^{\text{emp}}( \hat{\btheta^{ui}},\hat{\bZ^{u}},\hat{\bZ^{i}} ).
\end{equation}
\end{proof}
We use the following three lemmas to prove the inequality of the Eq.~\eqref{e:rtempdiff}. 

\begin{lemma}
\label{l:emp_tr}
At convergence, the training error of the learning algorithms is equal to the empirical error of their returned hypothesis:
\begin{equation}
R^{\text{emp}}_{\text{Inp-NCF}} = E_{\text{InP-NCF}}(\btheta^*,\bU^*,\bV^*), \quad  \quad R^{\text{emp}}_{\text{NCF}} = E_{\text{NCF}}(\hat{\btheta}).
\end{equation}
\end{lemma}
\begin{proof}
The proof is in the supplementary material.
\end{proof}

\begin{lemma}
\label{l:trerr}
Training error  of InP-NCF in Eq.~\eqref{e:ours}, which uses alternating optimization (explained in the supplementary material), is smaller than or equal to the  NCF in Eq.~\eqref{e:R}, i.e., $E_{\text{InP-NCF}}(\btheta^*,\bU^*,\bV^*) \le E_{\text{NCF}}(\hat{\btheta})$.
\end{lemma}
\begin{proof}
The proof is in the supplementary material.
\end{proof}

\begin{lemma}
\label{l:emperr}
At convergence, empirical error of the hypothesis returned by the InP-NCF is smaller than or equal to the one returned by NCF:
\begin{equation}
R^{\text{emp}}_{\text{Inp-NCF}} \le R^{\text{emp}}_{\text{NCF}} 
\end{equation}
\end{lemma}
\begin{proof}
The lemma holds because 1) from Lemma~\ref{l:emp_tr}, we have that the value of the empirical error is the same as the training error, and 2) from Lemma~\ref{l:trerr}, we have that training error of the InP-NCF is smaller than the training error of the NCF.
\end{proof}

In the supplementary material section~\ref{s:sup}, we have put the details of the equations and discussions on how to achieve the exact form of generalization bound.

\section{Experiments}
\label{s:exp}
\subsection{Experimental setup}

\paragraph{Datasets.}
For rating prediction, we use four datasets: 1) ml100k \citep{Harper15} dataset, which contains $100\,000$ ratings from around $1\,000$ users on $1\,600$ movies, 2) ml1m \citep{Harper15} dataset, which contains $1$ million ratings of around $6\,000$ users on $4\,000$ movies, 3) Amazon review data (Grocery and Gourmet Food) \citep{He16a} dataset, which contains $508\,800$ Amazon's Grocery and Gourmet Food product review from $86\,400$ users on $108\,500$ items, and 4) Ichiba dataset, which contains $1.5$ million Rakuten Ichiba\footnote{\url{https://rit.rakuten.com/data_release/}}'s product review from $324\,000$ users on $294\,000$ items. In each of the datasets, the ratings/reviews are integer numbers between $1$ and $5$.

For each dataset, we randomly select $80\%$ of the ratings as the training set, $10\%$ as the validation set, and $10\%$ as the test set. We repeat the process three times to create three training/validation/test sets. For each evaluation metric, we report the mean of each method on these three sets. The supplementary material section~\ref{s:sup} contains the standard deviation of the methods and a set of additional experiments that are not listed here.

For the implicit feedback prediction, we use two datasets: 1) ml1m \citep{Harper15} dataset, with the same number of users, items, and interactions as the one in the rating prediction task and 2) Amazon music (AMusic)\citep{He16a}, which contains $1\,700$ users, $13\,000$ items, and $46\,000$  ratings. Implicit feedback is achieved by turning all ratings to $1$.
Similar to \citep{Dong19, He17}, we use the leave-one-out evaluation, where the latest interaction of each user is held-out in the test set and the rest are considered as the training set.

\begin{table}
\caption{Number of neural network (NN) and input parameters in two network architectures used in Table~\ref{t:ml100L} and Table~\ref{t:valtr}.}
	\label{t:numP}
	\centering
			\begin{tabular}{@{}c@{\hspace{1ex}}c@{\hspace{1ex}}c@{\hspace{1ex}}c@{\hspace{3ex}}c@{\hspace{1ex}}c@{\hspace{1ex}}c@{}}
	\toprule
		& \multicolumn{3}{c}{$\bg^{u},\bg^{i},f=[1,1,2]$} & \multicolumn{3}{c}{$\bg^{u},\bg^{i},f=[3,3,4]$}\\
	& ml100k & ml1m & Ichiba & ml100k & ml1m & Ichiba\\
	\midrule
	 NN & $283K$  & $3.0M$ & $61.9M$  & $1.7M$ & $11.5M$ & $62.0M$ \\
	 \midrule
	Input & $160K$  & $1.6M$ & $2.4M$  & $160K$ & $1.6M$ & $2.4M$ \\
	 \bottomrule
	 	\end{tabular}
\end{table}
\vspace{-2ex}
\paragraph{Evaluation metrics.}  We report the root mean square error (RMSE) and average precision to evaluate the rating prediction performance. In RMSE we measure the error between the predicted and  actual rating. In precision, for each user, we retrieve the top p\%  of the items with the highest predicted ratings and compare them with the actual top p\% of items. We compute the precision for each user and report the average score.

We report Normalized Discounted Cumulative Gain (NDCG) and Hit Ratio (HR) to evaluate the implicit feedback prediction performance. We follow the same protocols as \citep{Dong19, He17} and truncate the rank list at $10$ for both metrics. HR measures whether the actual test item exists in the top-ranked list. NDCG is a measure of ranking quality and gives higher scores to the hits at top positions in the ranked list. We compute the HR and NDCG for each user and report the average score.

We have put the details and equations of all the evaluation metrics in the supplementary material section~\ref{s:sup}.
\vspace{-2ex}
\paragraph{Implementation details.}
We implemented our method using Keras with TensorFlow 2 backend. We ran all the experiments on a $12$GB GPU. For each method, we tried a set of activation functions (ReLU, SELU, and TanH), a range of learning rates and regularization parameters from $10^{-1}$ to $10^{-5}$, a set of optimizers (Adam, SGD, and RMSprop), and picked the one that works best. The supplementary material section \ref{s:sup} contains the details of the experimental setting.
\vspace{-2ex}
\paragraph{Number of network and input parameters.} 
In our tables, the notation $\bg^{i},\bg^u,f = [x,y,z]$ means that the three MLPs  $\bg^{i}$, $\bg^{u}$, and $f$ contain $x$, $y$, and $z$ hidden layers, respectively. Table~\ref{t:numP} contains the number of network and input parameters of two network architectures in the three datasets.  We call the network with smaller parameters ($\bg^{i},\bg^u,f = [1,1,2]$) shallow and the one with the larger number of parameters ($\bg^{i},\bg^u,f = [3,3,4]$) deep.

The number of network parameters is mainly determined by the size of the first layer of $\bg^u$ and $\bg^i$, which we called the implicit embedding layer. The size of this layer is proportional to the number of users and items in the dataset. We can see in Table~\ref{t:numP} that the number of neural network parameters in Ichiba is significantly more  than ml1m. The reason is that the number of users and items in Ichiba is around $60$ times larger than ml1m. We can also see that the number of inputs is always smaller than the network parameters. The difference between the input and network parameters is huge in Ichiba. This indicates that the main gain of our model is not from the number of parameters, but from the better generalization with learned input

\subsection{Experimental results}
\paragraph{Rating prediction: learning the input improves the performance.} In our experiments, NCF stands for neural collaborative filtering with the fixed interaction vector as the input, where the objective function can be found in Eq.~\eqref{e:R} and the architecture can be found in Fig.~\ref{f:inputs}. Our method is denoted by InP-NCF.

We try two different values to initialize the inputs: 1) we use the ratings as the input and denote the methods NCF (R) and InP-NCF (R) and 2) we turn all ratings  into $1$ (making them implicit) and denote the methods by NCF (I) and InP-NCF (I). These changes do not affect the outputs and in both cases the goal is to predict the ratings.

In Table~\ref{t:ml100L}, we compare RMSE and precision of InP-NCF with NCF on ml100k, ml1m, and Ichiba datasets. By comparing NCF (R) with InP-NCF (R) and NCF (I) with InP-NCF (I), we can see that learning the input consistently improves the results and gives lower RMSE and higher precision. 

The results of Table~\ref{t:ml100L} show that, even when we fix the inputs, explicit ratings are not always the best inputs for NCF models. In multiple cases, specifically in the Ichiba dataset, NCF with the implicit feedback as the input outperforms NCF with the ratings as the input.

Note that our objective function is highly nonconvex so different input initializations will lead to different results. While in both ml1m and ml100k datasets explicit rating initialization leads to slightly better results, it is in the Ichiba dataset where implicit feedback provides better results.

\begin{table}
	\caption{Rating prediction. We compare InP-NCF (learning the input and parameters) and NCF (fixed input).  Our approach consistently achieves better results. We report the mean of three runs. Supplementary material section \ref{s:sup} contains standard variation of methods and more results.}
	\label{t:ml100L}
	\centering
	\begin{tabular}{@{}c@{\hspace{1ex}}c@{\hspace{1ex}}c@{\hspace{1ex}}c@{\hspace{2ex}}@{\hspace{1ex}}c@{\hspace{1ex}}c@{\hspace{1ex}}c@{}}
	& \multicolumn{6}{c}{\dotfill RMSE \dotfill}  \\
	\toprule
		& \multicolumn{3}{c}{$\bg^{u},\bg^{i},f=[1,1,2]$} & \multicolumn{3}{c}{$\bg^{u},\bg^{i},f=[3,3,4]$}\\
	& ml100k & ml1m & Ichiba & ml100k & ml1m & Ichiba\\
	\midrule
	 NCF (R) & $0.903$  & $0.858$ & $0.883$  & $0.897$ & $0.856$ & $0.884$ \\
	 \midrule
	 	 InP-NCF (R) & $\mathbf{0.892}$ & $\mathbf{0.845}$ & $\mathbf{0.871}$  & $\mathbf{0.894}$ & $\mathbf{0.845}$  & $\mathbf{0.873}$  \\
	 \midrule
	 NCF (I) & $0.905$  & $0.860$ & $0.875$  & $0.896$ & $0.854$ & $0.874$\\ 
	 \midrule
	 InP-NCF (I)& $\mathbf{0.892}$  & $\mathbf{0.846}$ & $\mathbf{0.867}$  & $\mathbf{0.894}$ & $\mathbf{0.844}$ & $\mathbf{0.863}$ \\
	 \bottomrule
	 \\[-2ex]
\end{tabular}
		\begin{tabular}{@{}c@{\hspace{1ex}}c@{\hspace{1ex}}c@{\hspace{1ex}}c@{\hspace{2ex}}@{\hspace{1ex}}c@{\hspace{1ex}}c@{\hspace{1ex}}c@{}}
		& \multicolumn{6}{c}{\dotfill precision \dotfill}  \\
	\toprule
		& \multicolumn{3}{c}{$\bg^{u},\bg^{i},f=[1,1,2]$} & \multicolumn{3}{c}{$\bg^{u},\bg^{i},f=[3,3,4]$}\\
	& ml100k & ml1m & Ichiba & ml100k & ml1m & Ichiba\\
	\midrule
	 NCF (R) & $\mathbf{69.8\%}$  & $69.4\%$ & $79.5\%$  & $69.3\%$ & $70\%$ & $79.4\%$ \\
	 \midrule
	 	 InP-NCF (R) & $\mathbf{69.8\%}$ & $\mathbf{70.3\%}$ & $\mathbf{79.8\%}$  & $\mathbf{70.0\%}$ & $\mathbf{70.4\%}$  & $\mathbf{79.6\%}$  \\
	 \midrule
	 NCF (I) & $69.4\%$  & $69.3\%$ & $79.5\%$  & $69.7\%$ & $69.8\%$ & $79.6\%$\\ 
	 \midrule
	 InP-NCF (I)& $\mathbf{70.0\%}$  & $\mathbf{70.0\%}$ & $\mathbf{79.8\%}$  & $\mathbf{69.8\%}$ & $\mathbf{70.2\%}$ & $\mathbf{80.1\%}$ \\
	 \bottomrule
	\end{tabular}
\end{table}

\begin{table}
	\caption{Implicit feedback prediction. We train three different neural network architectures proposed in \cite{Dong19} with and without learning the inputs. The input learning strategy consistently improves the results. Our method also performs better than NAIS \cite{He18a}, which is an attention-based CF method.}
	\label{t:comp_imp}
		\begin{center}
			\begin{tabular}[c]{ccccc} 
				\toprule
				& \multicolumn{2}{c}{ml1m} & \multicolumn{2}{c}{ Amazon music} \\
 				method & HR & NDCG & HR & NDCG\\
				\midrule
				CFNet-rl & $70.7\%$ & $43.1\%$ & $39.2\% $ & $24.3\%$ \\
				\midrule
				InP-CFNet-rl & $\mathbf{71.5}\%$ & $\mathbf{43.3}\%$  & $\mathbf{40.0}\%$ & $\mathbf{24.6}\%$ \\
				\midrule
				CFNet-ml & $70.4\%$ & $42.5\%$  & $40.1\% $ & $24.5\%$ \\
				\midrule
				InP-CFNet-ml & $\mathbf{71.2}\%$ & $\mathbf{43.3}\%$ &$\mathbf{41.0}\% $ & $\mathbf{25.2}\%$\\
				\midrule
				CFNet  &  $70.4\%$& $42.8\%$  &  $39.0\% $ & $24.6\%$\\
				\midrule
				InP-CFNet  & $\mathbf{71.8}\%$ & $\mathbf{43.7}\%$&  $\mathbf{39.7}\% $ & $\mathbf{25.1}\%$ \\
				\midrule
				NAIS  & $68.4\%$ & $40.6\%$&  $35.2\% $ & $20.1\%$ \\
				
				\bottomrule
			\end{tabular}
		\end{center}
\end{table}

\paragraph{Implicit feedback prediction: learning the input improves the performance.}
In \cite{Dong19}, three different neural network architectures are used to create the representations in NCF. These architectures are called CFNet-rl, CFNet-ml, and CFNet. CFNet-rl uses MLPs  on top of user/item input interaction vectors to generate user/item representation and uses element-wise product to combine them. CFNet-ml uses a single linear layer to generate representations and uses concatenation to combine them. The CFNet architecture combines the representations of the CFNet-rl and CFNet-ml. Note that the input is fixed to the implicit interaction vector and the goal is to predict the implicit feedback. We ran these three models using the authors' code, which is publicly available \footnote{\url{https://github.com/familyld/DeepCF}}. The preprocessed datasets are also available within the code. We have put the HR and NDCG of these models in Table~\ref{t:comp_imp}.

To learn the input parameters in addition to the model parameters, we slightly changed the original code without changing the hyper-parameters. As we can see in Table~\ref{t:comp_imp}, the input learning strategy consistently improves the HR and NDCG of all the architectures.

We provide more details about CFNet and how our approach can improve its results in the supplementary material section \ref{s:sup}.

\paragraph{Comparison with the attention-based CF.} In Table~\ref{t:comp_imp}, we have put the results of NAIS and compared it with NCF and InP-NCF. We ran NAIS using the authors' code, which is publicly available \footnote{\url{https://github.com/AaronHeee/Neural-Attentive-Item-Similarity-Model}}. As we can see, InP-NCF achieves better results than NAIS on both datasets.

\begin{table}
	\caption{Training and validation RMSE of InP-NCF and NCF at different epochs. Our approach achieves lower RMSE using both deep and shallow networks.}
	\label{t:valtr}
	\centering
	\begin{tabular}{@{}c@{\hspace{1ex}}c@{\hspace{1ex}}c@{\hspace{1ex}}c@{\hspace{1ex}}c@{\hspace{1ex}}c@{\hspace{1ex}}c@{\hspace{1ex}}c@{}}
	\multicolumn{8}{c}{\dotfill ml1m, shallow network \dotfill}  \\
	\toprule
	\# epochs & $0$ & $2$ & $4$  & $6$ & $10$ & $15$ & best\\
	\midrule
	 NCF train& $1.61$ & $0.858$ & $0.835$ &  $0.824$ & $0.792$ & $0.764$  \\
	  NCF val & $1.60$ & $0.880$ & $0.867$ &  $0.867$  & $0.861$ & $0.872$ & $0.861$  \\
	 \midrule
	  InP-NCF train & $1.61$  & $0.849$ & $0.834$ & $0.798$ & $0.767$ & $0.732$ \\
	 InP-NCF val & $1.60$ &  $0.868$ & $0.876$ &  $0.853$ & $0.858$ & $0.864$ & $0.853$\\
	\bottomrule
	\\[-1ex]
	\end{tabular}
	\begin{tabular}{@{}c@{\hspace{1ex}}c@{\hspace{1ex}}c@{\hspace{1ex}}c@{\hspace{1ex}}c@{\hspace{1ex}}c@{\hspace{1ex}}c@{\hspace{1ex}}c@{}}
	\multicolumn{8}{c}{\dotfill ml1m, deep network \dotfill}  \\
	\toprule
	\# epochs &  $0$ & $2$ & $4$ & $6$ & $10$ & $15$ & best\\
	\midrule
	 NCF train&  $2.29$ & $0.870$ & $0.834$ & $0.816$ & $0.790$ & $0.753$\\
	  NCF val &  $2.27$ & $0.885$ & $0.868$ & $0.859$ & $0.861$ & $0.865$ & $0.857$ \\
	 \midrule
	  InP-NCF train & $2.29$ & $0.865$ & $0.820$ & $0.811$ & $0.763$& $0.721$ \\
	 InP-NCF val & $2.27$ & $0.887$ & $0.860$ & $0.859$ & $0.852$ & $0.865$ & $0.852$\\
	\bottomrule
	\end{tabular}
\end{table}

\begin{table}
\caption{We report RMSE and precision of InP-NCF using end-to-end optimization and alternating optimization. The results seem very similar regardless of the optimization approach. We run each experiments three times and report mean and standard deviation.}
	\label{t:altopt}
	\centering
	\begin{tabular}{cccc}
	\toprule
	& \multicolumn{3}{c}{\dotfill RMSE  \dotfill} \\
	& ml100k & ml1m & Ichiba \\
	\midrule
	 end-to-end & $0.892  \pm  0.003$ & $0.845  \pm  0.001$ & $0.871  \pm  0.004$ \\
	 \midrule
	 alt-opt & $0.893  \pm  0.001$ & $0.845  \pm  0.000$ & $0.869 \pm  0.005$\\[1ex]
	 \end{tabular}
	 \begin{tabular}{cccc}
	\toprule
	& \multicolumn{3}{c}{\dotfill precision  \dotfill} \\
	& ml100k & ml1m & Ichiba \\
	\midrule
	 end-to-end & $69.8\%  \pm  0.5$ & $70.3\%  \pm  0.09$ & $79.8\%  \pm  0.3$ \\
	 \midrule
	 alt-opt & $69.9\%  \pm  0.5$ & $70.2\%  \pm  0.1$ & $79.5\%  \pm  0.2$\\
	 \bottomrule
	 \end{tabular}
	 \vspace{-2ex}
\end{table}

\vspace{-2ex}
\paragraph{Input learning with a shallow network outperforms fixed input with a deep network.} In Table~\ref{t:ml100L}, we report RMSE of InP-NCF and NCF using a shallow network ($\bg^{i},\bg^u,f = [1,1,2]$) and a deep network ($\bg^{i},\bg^u,f = [3,3,4]$). In the ml100k dataset, the results get slightly better as we make the networks deeper. In ml1m, which has more training data, the improvements are more significant than ml100k, as we make the networks deeper. In Ichiba, there is no improvement by making the network deeper. 

An important point about the results of Table~\ref{t:ml100L} is that our approach InP-NCF with a shallow network achieves better results than NCF (fixed input) with a deep network. For example, our method achieves RMSE of $0.845$ in ml1m dataset using the shallow architecture, while NCF with fixed input achieves RMSE of $0.856$ with the deep network. 

To find the reason, we report RMSE of both methods on the training and validation sets using the shallow and the deep network on ml1m dataset in Table~\ref{t:valtr}. The first six columns contain RMSE from initialization to epoch $15$ and the last column contains the best RMSE on the validation set. There are three important points in these results. First, we can see that even using a shallow network, overfitting happens after $10$ epochs. That is why the methods do not gain much improvement by increasing the number of layers (i.e., number of parameters). Second, at epoch $15$, the training error of InP-NCF using the shallow network is smaller than the training error of NCF using the deep network. This means that while the number of input variables is smaller than the number of model parameters in ml1m, input variables are more important in decreasing the error. Finally, the results of Table~\ref{t:valtr} confirm our theoretical results. In both cases of deep and shallow networks, InP-NCF achieves a lower empirical error and a better validation error than NCF.
\vspace{-2ex}
\paragraph{Optimization.} We optimize our InP-NCF approach in two different ways: 1) the end-to-end optimization, which is followed by a post-input optimization, and 2) alternating optimization over the inputs and parameters. Table~\ref{t:altopt} shows the results on three dataset. We can see that the RMSE and precision of InP-NCF is very similar regardless of the optimization approaches.

\vspace{-2ex}
\paragraph{Post-optimization of inputs.} In our experiments, we observed that model parameters and input parameters overfit at different rates. Based on this observation, we first optimize all the parameters together until we see overfitting. Then, as a second step, we fix the parameters of the model and continue optimization over the inputs for a few more epochs.

In Table~\ref{t:valtr}, we have shown the training and validation RMSE of InP-NCF for $15$ epochs. We continue the optimization only over the input parameters for a few more epochs and put the results in Table~\ref{t:post_inp}. As we can see, this strategy decreases both training and validation RMSEs.

\begin{table}
	\caption{Post-optimization over the input variables. We fix the model parameters and continue the optimization of the InP-NCF models of Table~\ref{t:valtr} over the inputs. }
	\vspace{0ex}
	\label{t:post_inp}
	\centering
	\begin{tabular}{@{}c@{\hspace{1.5ex}}c@{\hspace{1.5ex}}c@{\hspace{1.5ex}}c@{\hspace{1.5ex}}c@{\hspace{1.5ex}}c@{}}
	\multicolumn{6}{c}{\dotfill ml1m, shallow network \dotfill}  \\
	\toprule
	\# epochs & $0$ & $1$ & $3$  & $5$ & $7$\\
	\midrule
	  InP-NCF train & $0.791$  & $0.782$ & $0.775$ & $0.771$ & $0.767$ \\
	 InP-NCF val & $0.853$ &  $0.847$ & $0.846$ &  $0.845$ & $0.845$\\
	\bottomrule
	\end{tabular}
	\begin{tabular}{@{}c@{\hspace{1.5ex}}c@{\hspace{1.5ex}}c@{\hspace{1.5ex}}c@{\hspace{1.5ex}}c@{\hspace{1.5ex}}c@{}}
	\multicolumn{6}{c}{\dotfill ml1m, deep network \dotfill}  \\
	\toprule
	\# epochs &  $0$ & $1$ & $3$ & $5$ & $7$\\
	 \midrule
	  InP-NCF train & $0.761$ & $0.749$ & $0.742$ & $0.737$ & $0.733$ \\
	 InP-NCF val & $0.852$ & $0.8473$ & $0.8470$ & $0.875$ & $0.848$\\
	\bottomrule
	\end{tabular}
	\vspace{-2ex}
\end{table}

\vspace{-2ex}
\paragraph{Comparison with the state-of-the-art methods.}
We compare our method with several baselines and recent works by reporting RMSE and precision in Table~\ref{t:comp1} and Table~\ref{t:comp3}. MF \cite{Koren09} and NeuMF \cite{He17} are collaborating filtering (CF) methods with the ID as the input. MF uses dot product and NeuMF uses MLPs to model the interaction between the user and item representations and to predict the ratings. Autorec \cite{Sedhain15} and DeepCF \citep{Dong19} are CF methods with the ratings as the input. Autorec uses autoencoders to reconstruct the ratings and DeepCF uses the MLPs to model user-item interaction. DHA \cite{Li18} and aSDAE \cite{Dong17} are two autoencoder-based hybrid methods, which try to reconstruct the side information and the ratings of the users and items. HIRE \cite{Liu19} is a hybrid method that considers the hierarchical user and item side information. DSSM \cite{Huang13} is a content-based recommender method, which uses deep neural networks to learn the representations.

NeuMF, DeepCF, and DSSM were originally proposed for the prediction of implicit feedback. To make these methods applicable to the explicit feedback prediction, we modify their objective function and use a mean squared error loss function. For the hybrid and content-based methods, we create bag of words (BoW) features from the user and item side information, such as age, gender, occupation, movie title, genre, etc. DSSM is not applicable in the Amazon dataset, as there is no user side information.

We got an out-of-memory (OOM) error in the training of HIRE, DHA, and aSDAE in Amazon and Ichiba datasets. HIRE's code (publicly available by the authors) needs the dense interaction matrix in training, which can not be stored in memory for Amazon and Ichiba datasets. DHA and aSDAE reconstruct the users' and items' interaction vectors. This makes the first and the last layer of their autoencoders huge in Amazon and Ichiba datasets and makes their model's size larger than the $12GB$ memory of our GPU.

\begin{table}[t]
	\caption{We compare RMSE of our approach, InP-NCF, with the hybrid and collaborative filtering methods. We put "OM" in the tables whenever we get an out-of-memory error. Our approach outperforms the rest of the methods. }
	\label{t:comp1}
		\begin{center}
			\begin{tabular}[l]{ccccc} 
				\toprule
				 method & ml100k & ml1m  & Amazon & Ichiba\\
				\midrule
MF & $0.906 \pm 0.002$ & $0.857 \pm 0.0005$ & $1.107 \pm 0.003$ & $0.876 \pm 0.003$ \\
\midrule
Autorec & $0.900 \pm 0.002$ &   $0.855 \pm 0.0003$ & $1.108 \pm 0.01$ &  $0.876 \pm 0.059$  \\
\midrule
NeuMF  & $0.948 \pm 0.005$ &  $0.886 \pm 0.001$ & $1.140 \pm 0.004$ &  $0.900 \pm 0.004$ \\
\midrule
DSSM & $0.934 \pm 0.002$ & $0.941 \pm 0.0004$ & NA & $0.913 \pm 0.003$ \\
\midrule
DHA   &  $0.939 \pm 0.002$  &   $0.865 \pm 0.001$ &  OM &   OM \\
\midrule
aSDAE   &  $0.946 \pm 0.005$ & $0.879 \pm 0.005$ & OM  &  OM  \\
\midrule
HIRE   &  $0.930 \pm 0.006$ & $0.861 \pm 0.004$ &  OM & OM  \\
\midrule
DeepCF  &  $0.900 \pm 0.002$ & $0.866 \pm 0.001$ &  $1.123 \pm 0.002$ & $0.886 \pm 0.003$  \\
\midrule
InP-NCF & $\mathbf{0.892 \pm 0.004}$   &  $\mathbf{0.845 \pm 0.001}$ & $\mathbf{1.101 \pm 0.003}$  &  $\mathbf{0.867 \pm 0.003}$ \\
				\bottomrule
			\end{tabular}
		\end{center}
\end{table}

\begin{table}[!t]
\caption{We report precision of our approach and the state-of-the-art methods. We report the precision by creating the relevant and retrieved sets using top $10\%$ and $25\%$ of the items. We put "OM" in the tables whenever we get an out-of-memory error. Our approach outperforms the rest of the methods. }
\label{t:comp3}
\begin{center}
\begin{tabular}[c]{ccccc} 
\toprule
& \multicolumn{2}{c}{ml1m} & \multicolumn{2}{c}{ Amazon review} \\
 method & top $10\%$ & top $25\%$ & top $10\%$ & top $25\%$ \\
\midrule
MF & $55.6\% \pm 0.16$ & $68.05\% \pm 0.45$ & $64.9\% \pm 0.04$ & $71.5\% \pm 0.67$ \\
\midrule
Autorec & $57.6\% \pm 0.26$ & $69.5\% \pm 0.43$  & $64.6\% \pm 0.98$ & $71.3\% \pm 0.62$ \\
\midrule
NeuMF & $56.8\% \pm 0.12$ & $68.9\% \pm 0.48$  & $66.8\%  \pm 0.30 $ & $72.6\% \pm 0.09$ \\
\midrule
DSSM & $54.7\% \pm 0.35$ & $67.2\% \pm 0.30$ & NA & NA\\
\midrule
DHA  &  $57.4\% \pm 0.78$& $69.3\% \pm 0.23$  &  OM & OM\\
\midrule
aSDAE  & $56.4\% \pm 0.39$ & $68.0\% \pm 0.42$&  OM& OM \\
\midrule
HIRE   & $57.4\% \pm 0.08$ & $69.4\% \pm 0.55$&  OM & OM  \\
\midrule
DeepCF  & $57.1\% \pm 0.16$ & $69.2\% \pm 0.29$&  $67.7\% \pm 0.53$ & $73.9\% \pm 0.2$  \\
\midrule
InP-NCF & $\mathbf{58.6\% \pm 0.04}$  & $\mathbf{70.1\% \pm 0.38}$ &   $\mathbf{69.3\%  \pm 0.34}$& $\mathbf{75.5\% \pm 0.00}$ \\
\bottomrule
\end{tabular}
\end{center}
\end{table}

We compare the RMSE (Table~\ref{t:comp1}) and the precision (Table~\ref{t:comp3}) on four datasets. Our method achieves significantly better results than the state-of-the-art methods in all the experiments. These results confirm the advantage of simultaneous learning of the inputs and parameters in NCF.

\section{Conclusion}
In this paper, we studied the role of the input in neural collaborative filtering. We showed that there is an implicit user/item embedding matrix in the first hidden layer of the MLPs, which takes the interaction vector as the input. The non-zero elements of this vector determine which embedding vectors to choose. The output of the first hidden layer is a weighted sum of the embedding vectors, where the weights are determined by the value of the non-zero elements. We proposed to learn the value of the non-zero elements of the interaction vector with the network parameters jointly to achieve more powerful representations. We theoretically analyzed our approach and proved that it achieves a better generalization bound than previous works. Extensive experiments on predicting implicit and explicit feedback on several datasets verified the effectiveness of our approach.

{\small
\bibliographystyle{plainnat}

}

\newpage
\section{Supplementary material}
\label{s:sup}
This supplementary material contains the following: 1) proof of the lemmas in the main script, 2) the exact form of the generalization bound of InP-NCF and NCF, 3) analysis on applying our input learning approach to the NCF framework of CFNet \citep{Dong19}, 4) an extension of our experiments, and 5) the experimental settings of each method in our paper, which includes number of layers, learning rate, optimizer, etc.

\subsection{Proof of the lemmas}
\label{s:risk}

We assume the learning algorithms return the following parameters at the convergence: 
\begin{equation}
{\btheta^{ui}}^*,{\bZ^{u}}^{*},{\bZ^{i}}^{*} = \mathcal{A}_{\text{InP-NCF}}(\mathcal{D}), \qquad \qquad  \hat{\btheta^{ui}},\hat{\bZ^{u}},\hat{\bZ^{i}} = \mathcal{A}_{\text{NCF}}(\mathcal{D}).
\end{equation}

We first prove that the empirical risks of the NCF and InP-NCF at convergence is the same as their training errors. Then, we prove that InP-NCF's training error is smaller than the NCF's training error.

Let us assume that NCF converges to the network parameters $\hat{\btheta}=[\hat{\btheta^{ui}},\hat{\btheta^{u}},\hat{\btheta^{i}}]$, where $\hat{\bZ^{u}} = \bg^{u}(\bR;\hat{\btheta^{u}})$ and $\hat{\bZ^{i}} = \bg^{i}(\bR;\hat{\btheta^{i}})$. Also, assume that InP-NCF converges to network parameters $\btheta^*=[{\btheta^{ui}}^*,{\btheta^{u}}^*,{\btheta^{i}}^*]$ and input parameters [$\bU^*,\bV^*$], where ${\bZ^{u}}^{*} = \bg^{u}(\bU^*;{\btheta^{u}}^{*})$ and ${\bZ^{i}}^{*} = \bg^{i}(\bV^*;{\btheta^{i}}^{*})$.

\begin{lemma}
\label{l:emp_tr_supp}
At convergence, the training error of the learning algorithms is equal to the empirical error of their returned hypothesis:
\begin{equation}
R^{\text{emp}}_{\text{Inp-NCF}} = E_{\text{InP-NCF}}(\btheta^*,\bU^*,\bV^*), \quad  \quad R^{\text{emp}}_{\text{NCF}} = E_{\text{NCF}}(\hat{\btheta}).
\end{equation}
\end{lemma}
\begin{proof}
We prove the lemma only for the InP-NCF, as the one for the NCF can be proved in a very similar way. Let us start by writing the empirical error of the InP-NCF using Eq.~\eqref{e:emp_risk} and Eq.~\eqref{e:zmodel}:
\begin{align}
\label{e:emp_exp_supp}
R^{\text{emp}}_{\text{InP-NCF}}  = R^{\text{emp}}( {\btheta^{ui}}^*,{\bZ^{u}}^*,{\bZ^{i}}^* )= \frac{1}{mn} \sum_{j=1}^{m}\sum_{k=1}^n  L(R_{jk},f(\bz_{jk};{\btheta^{ui}}^*)))\\ \nonumber
 \text{s.t.} \quad \bz_{jk} = [\bz^u_j,\bz^i_k]  \quad   \bz^{u}_j = \bI^{u}_j{\bZ^{u}}^* , \quad \bz^{i}_k = \bI^{i}_k{\bZ^{i}}^* 
\end{align}
The main term of the Eq.~\eqref{e:emp_exp_supp} is the same as the main term of the NCF objective in  Eq.~\eqref{e:R}, which is also the main term of the InP-NCF objective function in Eq.~\eqref{e:ours}. But, the constraints seem to be different. We rewrite the constraints of the  Eq.~\eqref{e:emp_exp_supp} as follows:
\begin{equation}
\label{e:nc}
\bz^{u}_j = \bI^{u}_j{\bZ^{u}}^* = \bI^{u}_j \bg^{u}(\bU^*;{\btheta^{u}}^*) = \bg^{u}({\bU}^{*}_{j,:};{\btheta^{u}}^*),  \quad \quad \bz^{i}_k = \bI^{i}_k{\bZ^{i}}^* =  \bI^{i}_k \bg^{i}(\bV^*;{\btheta^{i}}^*) = \bg^{i}(\bV^{*}_{k,:};{\btheta^{i}}^*).
\end{equation}
By replacing the constraints of the Eq.~\eqref{e:emp_exp_supp} with the ones in  Eq.~\eqref{e:nc}, we can see that both the main terms and constraints of the Eq.~\eqref{e:emp_exp_supp} and Eq.~\eqref{e:ours} are the same, which means that $R^{\text{emp}}_{\text{Inp-NCF}} = E_{\text{InP-NCF}}(\btheta^*,\bU^*,\bV^*)$ at convergence. We can prove this fact for NCF in the same way.
\end{proof}

Since the value of the training error is the same as the empirical error, we 
 focus on comparing the training errors of the InP-NCF and NCF, instead of the empirical errors. 
 
Let us assume that both InP-NCF and NCF converge to their optimal solutions. In other words, assume that $\hat{\btheta}$ is the optimal solution of NCF and [$\btheta^*,\bU^*,\bV^*$] is the optimal solution of InP-NCF. 

\begin{lemma}
\label{l:trerr_optimal}
At the optimal solution, training error  of InP-NCF is smaller than or equal to the  NCF, i.e., $E_{\text{InP-NCF}}(\btheta^*,\bU^*,\bV^*) \le E_{\text{NCF}}(\hat{\btheta})$.
\end{lemma}
\begin{proof}
NCF fixes the input to the ratings ($\bU = \bR$ and $\bV= \bR^T$) and learns the network parameters. Since [$\btheta^*,\bU^*,\bV^*$] is the optimal solution of InP-NCF, we have:
\begin{equation*}
\label{e:trproof_optimal}
E_{\text{NCF}}(\hat{\btheta}) = E_{\text{InP-NCF}}(\hat{\btheta},\bR,\bR^T) \ge E_{\text{InP-NCF}}(\btheta^{*},\bU^{*},\bV^{*}).
\end{equation*}
\end{proof}
 
In practice, the objective functions are highly non-convex and the models will not converge to the optimal solutions. So we have to compare the empirical risks of the models at their local minima. With our end-to-end optimization of the parameters and the inputs, we do not have enough information about the the convergence point of InP-NCF and how it is compared with the NCF. To ensure that InP-NCF achieves a lower training error than NCF, no matter how we choose the functions $f()$, $\bg^{u}()$ and $\bg^{i}()$, we have designed the following alternating optimization algorithm. 

We apply alternating optimization over the input parameters ($\bU$,$\bV$) and the network parameters $\btheta$ to minimize the objective function of Eq.~\eqref{e:ours}. We start the optimization from $\btheta$ since it is initialized randomly, as opposed to the input which can be initialized by the interaction vector. Formally, we initialize $\bU^{0} = \bR$ and $\bV^{0} = \bR^T$ and repeat the following steps until convergence ($t\ge 1$):
\begin{equation}
\label{e:altopt}
\text{(a):}\ \btheta^{t} = \argmin_{\btheta} E_{\text{InP-NCF}}(\btheta,\bU^{t-1},\bV^{t-1}), \quad
\text{(b):}\ \bU^{t},\bV^{t} = \argmin_{\bU,\bV} E_{\text{InP-NCF}}(\btheta^{t},\bU,\bV)  
\end{equation}

\begin{lemma}
\label{l:trerr_supp}
Training error  of InP-NCF in Eq.~\eqref{e:ours}, which uses alternating optimization of Eq.~\eqref{e:altopt}, is smaller than or equal to the  NCF in Eq.~\eqref{e:R}, i.e., $E_{\text{InP-NCF}}(\btheta^*,\bU^*,\bV^*) \le E_{\text{NCF}}(\hat{\btheta})$.
\end{lemma}
\begin{proof}
NCF fixes the input to the ratings ($\bU = \bR$ and $\bV= \bR^T$) and learns the network parameters. This is equivalent to step (a) in Eq.~\eqref{e:altopt} at the first iteration ($t=1$) of InP-NCF. In other words, the training error of NCF can be written as $E_{\text{NCF}}(\hat{\btheta}) = E_{\text{InP-NCF}}(\btheta^{1},\bU^{0},\bV^{0})$. InP-NCF decreases the training error at each step of Eq.~\eqref{e:altopt}, so we have:
\begin{equation}
\label{e:trproof}
E_{\text{NCF}}(\hat{\btheta}) = E_{\text{InP-NCF}}(\btheta^{1},\bU^{0},\bV^{0})  \ge  E_{\text{InP-NCF}}(\btheta^{1},\bU^{1},\bV^{1}) \ge  \dots \ge E_{\text{InP-NCF}}(\btheta^{*},\bU^{*},\bV^{*}).
\end{equation}
\end{proof}
Note that we only need the alternating optimization of InP-NCF for our theoretical analysis. In our experiments, we will show that the performance of the end-to-end optimization and alternating optimization of InP-NCF generate similar results.

\begin{lemma}
\label{l:emperr_supp}
At convergence, empirical error of the hypothesis returned by the InP-NCF is smaller than or equal to the one returned by NCF:
\begin{equation}
R^{\text{emp}}_{\text{Inp-NCF}} \le R^{\text{emp}}_{\text{NCF}} 
\end{equation}
\end{lemma}
\begin{proof}
The lemma holds because 1) in Lemma~\ref{l:emp_tr_supp} we proved that the value of the empirical error is the same as the training error, and 2) in Lemma~\ref{l:trerr_supp} we proved that training error of the InP-NCF is smaller than the training error of the NCF.
\end{proof}

\subsection{Exact form of the generalization bound of InP-NCF and NCF}
As mentioned in the main script, the prediction model of the latent factor-based collaborative filtering methods (using concatenation to combine the representations and an MLP $f()$ to predict the interaction) can be written as: 
\begin{equation}
	\label{e:zmodel_supp}
	\hat{R}_{jk} =  f(\bz_{jk};\btheta^{ui}) \quad \text{s.t.} \quad
	\bz_{jk}=[\bz^u_j,\bz^i_k], \   \bz^{u}_j = \bI^{u}_j\bZ^{u}, \ \bz^{i}_k = \bI^{i}_k\bZ^{i},
\end{equation}
where $ \bI^{u}_j$ and $\bI^{i}_k$ are the one-hot encoding vectors of the $j$th user and $k$th item. The input to this model is the user and item IDs and the parameters are $\btheta^{ui}$, $\bZ^u$, and $\bZ^i$.

Here, the goal is to write the exact form of the generalization bound for our model in Eq.~\eqref{e:zmodel_supp}. To achieve this, we make a few assumptions. First, let us assume that the joint representation $\bz^{ui}$ is achieved by the summation of the user and item representations. In that case, we can compute $\bz^{ui}$ by as single vector-matrix multiplication:
\begin{equation}
	\bz^{ui}_{jk} = \bI^{ui}_{jk} \bZ^{ui} \quad \text{s.t.} \quad  \bI^{ui}_{jk}  = [\bI^{u}_j, \bI^{i}_k], \ \bZ^{ui} = [\bZ^{u},\bZ^{i}],
\end{equation}
where $\bZ^{ui}$ is achieved by appending $\bZ^{i}$ to the end of $\bZ^{u}$. 

The second assumption changes the objective function of NCF from regression or binary cross-entropy losses to the hinge loss. Let us assume there are $p$ unique interaction values, i.e., $R_{jk} \in\{1,2,\dots,p\}$. In case of implicit feedback, we have $p=2$. In case of explicit feedback, $p$ is the number of unique rating values. We define $\bbf()$ as an MLP with $L_f$ hidden layers and $d$ units per layer, which takes $\bz_{jk}$ as the input and generates a $p$-dimensional vector as the output. The output of $\bbf()$ is a score for each of the interaction values, where the interaction with the maximum score will be selected as the predicted interaction. We define the objective function using the hinge loss with margin of $\gamma$ as follows:
\begin{multline}
E_{\text{NCF}}(\bz_{jk},\btheta^{ui}) = \frac{1}{N} \sum_{j=1}^{m}\sum_{k=1}^n \ind(R_{jk}>0)
 \max(0,\gamma+\max_{y \ne R_{jk}}(\bbf(\bz_{jk};\btheta^{ui})[y]) - \bbf(\bz_{jk};\btheta^{ui})[R_{jk}]),
\end{multline}
where $N$ is the number of non-zero elements of the interaction matrix $\bR$ and $\bbf(\bz_{jk};\btheta^{ui})[y]$ returns the value of the $y$th index  of $\bbf(\bz_{jk};\btheta^{ui})$. By setting $\gamma=0$, the hinge loss counts the number of misclassified interactions. 

Based on the above assumptions, the number of learnable parameters in the model of Eq.~\eqref{e:zmodel_supp} is 
\begin{equation}
\label{e:nump1}
	\pi = (n+m+p)\times d + (L_{f}-1)(d\times d),
\end{equation}
where $p$ is the number of unique interaction values and $d$ is the dimension of the representations in $\bZ^u$ and $\bZ^{i}$ and the number of hidden units in each layer of the MLP $\bbf()$.

Note that the learnable parameters include $\bZ^u \in \bbR^{m \times d}$, $\bZ^{i}\in \bbR^{n \times d}$, and the weights of the MLP $\bbf()$, which contains $d\times d$ weight matrix in each of the first $L_{f}$ layers and $p\times d$ weight matrix in the last (output) layer.

Because of the above assumptions, the rating of the user $j$ on item $k$ can be estimated by a feedforward neural network, where the input is $\bI^{ui}_{jk}$, the weight matrix of the first layer is $\bZ^{ui}$, and the layers of $f()$ construct the second to the last layers. \citet{Bartlett19} gives the following generalization bound for our feedforward network with ReLU activation functions:
\begin{equation}
\label{e:bound1}
E^{exp} \le E^{tr} + \tilde{O}(\sqrt{\frac{(L_{f}+1)\pi}{N}}),
\end{equation}
where $\tilde{O}$ is the upper bound of the complexity up to a logarithmic factor, $\pi$ is the number of parameters defined in Eq.~\eqref{e:nump1}, and $N$ is the number of samples (non-zero interactions).

The bound in Eq.~\eqref{e:bound1} is linear in the number of parameters of the neural network, which becomes weak as the network becomes deep. This bound explains why InP-NCF outperforms NCF using shallow neural networks, which have one or two hidden layers and a few hundred units per layer.  For deeper networks, we could use the bound proposed by \citet{Neyshabur18} to show that InP-NCF has a smaller upper bound for generalization error than NCF. This bound depends on the norm of the weights of network, which makes it more suitable for deeper networks.

\subsection{Input learning with the NCF framework of CFNet \citep{Dong19}}
Let us first review the NCF framework of the CFNet \cite{Dong19} and see how it predicts the implicit feedback. CFNet learns two types of joint user-item representations. The first one is achieved by applying MLPs to the user and item representations and computing their element-wise product:
\begin{equation}
\label{e:cf1}
\bz_{jk}^{(1)} = \bz^{u}_j \odot \bz^{i}_k \quad \text{s.t.}  \quad
\bz^{u}_j = \bg^{u}(\bR_{j,:};\btheta^{u}_L), \quad \bz^{i}_k = \bg^{i}(\bR_{:,k};\btheta^{i}_L),
\end{equation}
where $\bg^{u}$ and $\bg^{i}$ are user and item MLPs, and $\odot$ is an element-wise product.

The second joint representation is generated in three steps: 1) computing linear transformation of the user and item interaction vectors, 2) concatenating the two output vectors, and 3) applying an MLP to the concatenation:
\begin{equation}
\label{e:cf2}
\bz_{jk}^{(2)} = \bh (\bv_{jk};\btheta^{h}_L) \quad \text{s.t.} \quad \bv_{jk} = [\bv_j,\bv_k], \quad \bv_j = \bR_{j,:} \bA^u, \quad \bv_k = \bR_{:,k}^T\bA^{i}
\end{equation}
where $\bA^u \in \bbR^{m\times d}$ and $\bA^i \in \bbR^{n\times d}$ are used for linear transformation and $\bh()$ is an MLP.

The final objective function can be written as:
\begin{equation}
\label{e:cfnet}
E_{\text{CFNet}}(\btheta^{u}_L,\btheta^{i}_L,\bA^u,\bA^i,\btheta^{h},\btheta^{ui})= \frac{1}{mn} \sum_{j=1}^{m}\sum_{k=1}^n  L(R_{jk},f([\bz_{jk}^{(1)},\bz_{jk}^{(2)}];\btheta^{ui})),
\end{equation}
where $\bz_{jk}^{(1)}$ and $\bz_{jk}^{(2)}$ are defined in Eq.~\eqref{e:cf1} and Eq.~\eqref{e:cf2}. In CFNet, the MLP $\bbf()$ contains a single nonlinear layer. 

\subsubsection{Input learning in CFNet} 
We define the same notations as the ones in our main script.  For the $j$th user, we define the set $N_j=\{l|R_{jl} \ne 0, 1 \le l \le n\}$, which contains the index of all the items interacted by the $j$th user. We define the set of user input vectors as $\bU \in \bbR^{m\times n}$, where $\bU_{j,:}$ is the input vector of the $j$th user. We learn $U_{jl}$ if $l \in N_j$ and fix it to $0$ otherwise.

For the $k$th item, we define the set $N_k=\{l|R_{lk} \ne 0, 1 \le l \le m\}$, which contains the index of all the users who interacted with the $k$th item. We also define the set of item input vectors as $\bV \in \bbR^{n \times m}$. We learn $V_{kl}$ if $l \in N_k$ and fix it to $0$ otherwise. 

We consider $\bP^{u}$ and $\bP^{i}$ as the weights of the first layer of the MLPs $\bg^{u}$ and $\bg^{i}$, respectively. Our objective function is defined as follows:
\begin{equation}
	\label{e:ours_cfnet}
	\begin{split}
	&E_{\text{InP-CFNet}}(\btheta,\bU,\bV) =  E_{\text{CFNet}}(\btheta^{u}_L,\btheta^{i}_L,\bA^u,\bA^i,\btheta^{h},\btheta^{ui}) \quad \text{s.t.} \\
	 &\bz^{u}_j = \bg^{u}(\by^{u}_j;\btheta^{u}_{L-1}), \quad
	\by^{u}_{j}=\sigma(\sum_{l \in N_j}  U_{jl} \bP^{u}_{l,:}), \quad \bv^{u}_{j}=\sigma(\sum_{l \in N_j}  U_{jl} \bA^{u}_{l,:}),\\
		&\bz^{i}_k = \bg^{i}(\by^{i}_k;\btheta^{i}_{L-1}), \quad \by^{i}_{k}=\sigma(\sum_{l \in N_k} V_{kl} \bP^{i}_{l,:}), \quad \bv^{i}_{k}=\sigma(\sum_{l \in N_k} V_{kl} \bA^{i}_{l,:})
		 \end{split}
\end{equation}
where $\btheta=[\btheta^{u}_{L-1},\bP^{u},\btheta^{i}_{L-1},\bP^{i},\bA^{u},\bA^{i},\btheta^{h},\btheta^{ui}]$ contains all the parameters of the CFNet network, $\bP^{u}$ and $\bP^{i}$ contain the weights of the first layer of the $\bg^{u}$ and $\bg^{i}$, respectively, and $\by^{u}_j$ and $\by^{i}_k$ are the output of the first layer of the $\bg^{u}()$ and $\bg^{i}()$, respectively. By setting $\smash{U_{jl} = R_{jl}}$ and $\smash{V_{kl}=R_{lk}}$, our objective in Eq.~\eqref{e:ours_cfnet} becomes equivalent to the previous CFNet's objective function in Eq.~\eqref{e:cfnet}.

\subsubsection{Theoretical analysis of InP-CFNet}
Similar to our main theorem in the main script, we can prove that InP-CFNet achieves a smaller generalization bound than CFNet. Since the details are very similar to our proof in the main script, and to avoid repetition, we review the main ideas.

 As mentioned in our main script, collaborative filtering (matrix completion) methods, including ours, assume that the model has seen all users and items and a small subset of their interactions in the training set. So, the model for both CFNet and InP-CFNet can be written as:
 \begin{equation}
 \label{e:test_cfnet}
 \begin{split}
&\hat{R}_{jk} = f([\bz_{jk}^{(1)},\bz_{jk}^{(2)}];\btheta^{ui}) \\
&\bz_{jk}^{(1)} = \bz^{u}_j \odot \bz^{i}_k \quad \text{s.t.}  \quad
\bz^{u}_j = \bI^{u}_{j}\bZ^{u}, \quad \bz^{i}_k = \bI^{i}_k \bZ^{i}\\
&\bz_{jk}^{(2)} = \bh (\bv_{jk};\btheta^{h}_L) \quad \text{s.t.} \quad \bv_{jk} = [\bv_j,\bv_k], \quad \bv_j = \bI^{u}_j\bV^{u}, \quad \bv_k = \bI^{i}_k\bV^{i}
\end{split}
 \end{equation}
The inputs of the above model are the user and item IDs, and the parameters are $\btheta^{ui}$, $\btheta^{h}_L$, $\bZ^{u}$, $\bZ^{i}$, $\bV^{u}$, and $\bV^{i}$. 

The key point here is that InP-CFNet and CFNet are both learning algorithms, which train their own objective function and return the parameters of the model in Eq.~\eqref{e:test_cfnet}. Similar to our theorem for the NCF and InP-NCF, the learning algorithm with the smaller empirical error gives a better bound. We can prove the the following facts in the same way we did for NCF and InP-NCF:
\begin{itemize}
\item At convergence, the empirical risk of the InP-CFNet and CFNet are the same as their training errors.
\item At convergence, InP-CFNet has a smaller training error than CFNet: $E^{\text{tr}}_{\text{Inp-CFNet}} \le E^{\text{tr}}_{\text{CFNet}}$.  We can  prove this using an alternating optimization over the inputs and network parameters of the InP-CFNet and initializing the inputs using the interaction values in the first step.
\end{itemize}

Since the InP-CFNet achieves a lower training error than CFNet, InP-CFNet has a smaller generalization bound than CFNet.

\begin{figure}[!t]
  \centering
  \begin{tabular}{@{}c@{\hspace{2ex}}c@{}}
  Initilization & After learning \\
	\includegraphics*[width=0.4\linewidth]{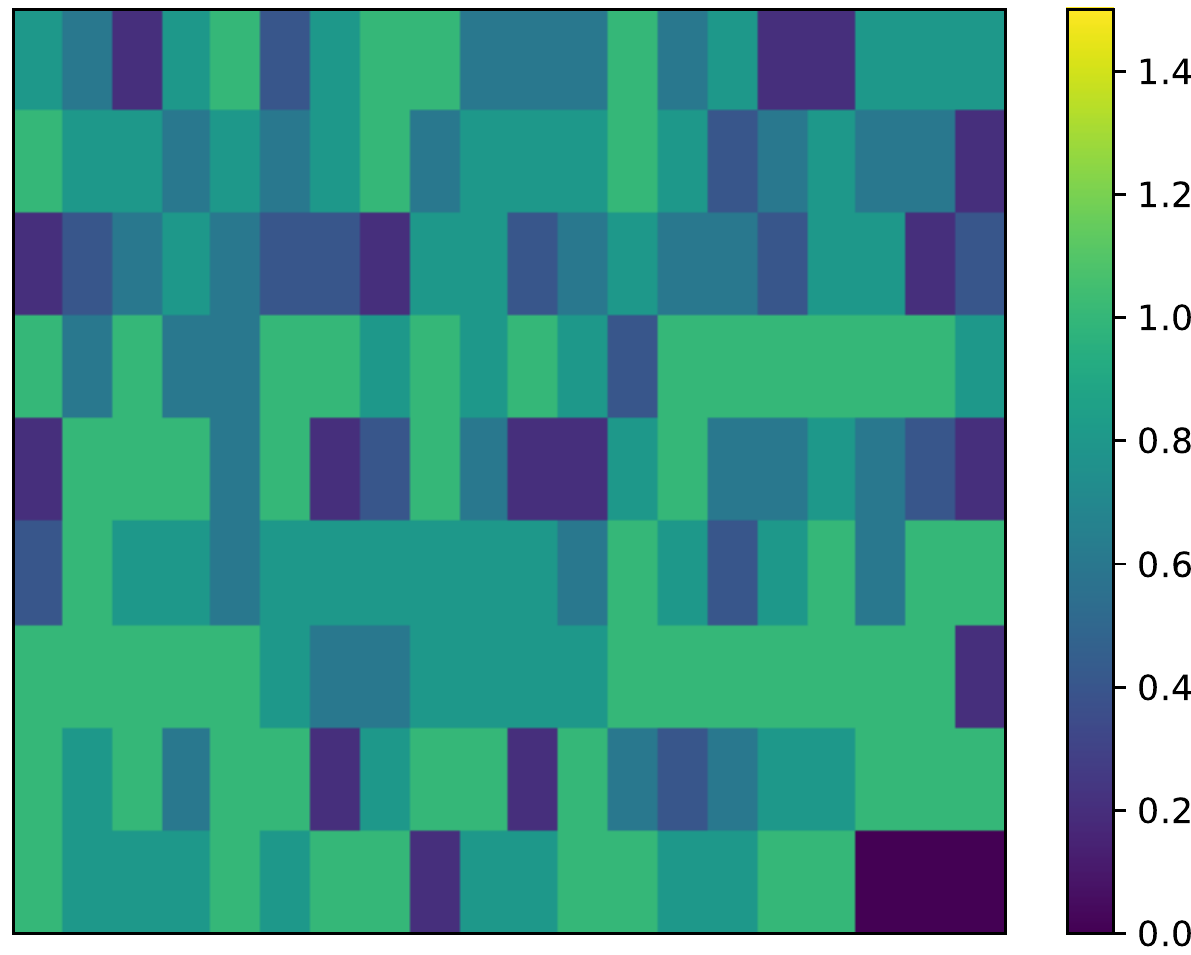}&
	 \includegraphics*[width=0.4\linewidth]{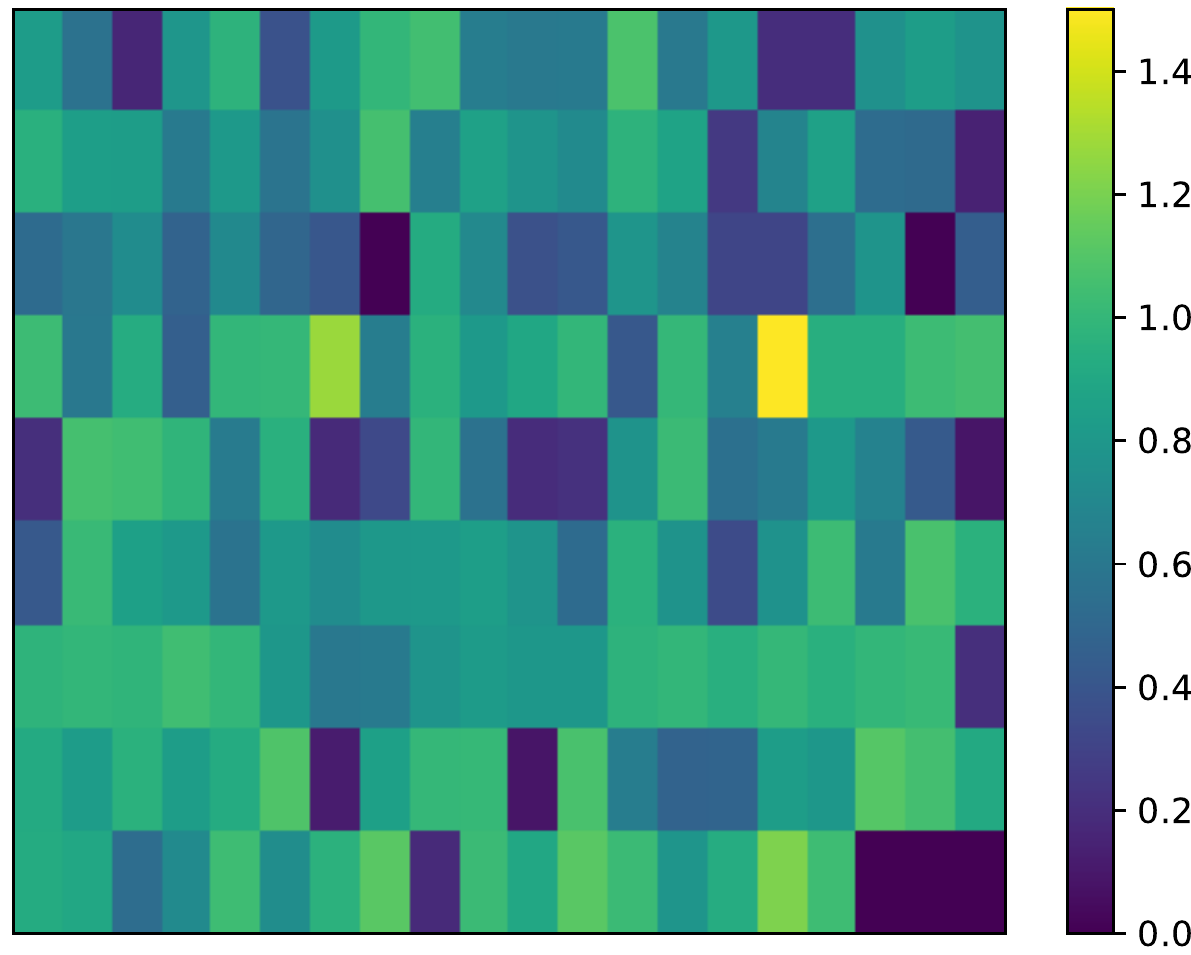}\\
  \end{tabular}
    \vspace{-2ex}
  \caption{Visualization of the input vectors of $10$ users in ml100k dataset. Each row corresponds to the ratings of one user.}
  \vspace{-3ex}
\label{f:visual}
\end{figure}

\subsection{Additional experimental results}
\subsubsection{Visualization of the values of the input vectors.} In Fig.~\ref{f:visual}, we visualize the values of the input vectors at the initialization (left panel) and after learning them with our approach (right panel), in ml100k dataset. We randomly select $10$ users and extract their ratings on $20$ items. 

Each row of the two matrices in Fig.~\ref{f:visual} corresponds to one user. At the initialization, the input vectors contain the ratings of the users on items. There are $5$ unique ratings (from $1$ to $5$) in ml100k dataset. We divide the ratings by $5$, so the input values are between $0.2$ to $1$. Note that there are a few $0$s at the initialization. That's because some of the selected users have rated less than $20$ items.

At the right panel of Fig.~\ref{f:visual}, we show the input values of those $10$ users after learning the input with the InP-NCF. We can see that there are many more unique values. Also, sometimes the input values change significantly from their initial values. To see this better, note that the range of input values at initialization is between $0$ and $1$. When we learn the input, some values will change to $1.4$ (look at the yellow colors in the right panel). We have set the range of the colorbar from $0$ to $1.4$ in both panels of the Fig~\ref{f:visual} to keep the colors consistent and to make the comparison easier.

\subsubsection{Evaluation metrics.}
We report the root mean square error (RMSE) and average precision to evaluate the rating prediction performance. RMSE is defined as:
\begin{equation}
\text{RMSE} =\textstyle \sqrt{\frac{1}{|T|} \sum_{R_{jk}\in T} (R_{jk} - \hat{R}_{jk})^2}
\end{equation}
where set $T$ contains all the ratings in the test set and $R_{jk} \in T$ and $\hat{R}_{jk}$ are the actual rating and predicted rating of the user $j$ on item $k$, respectively. 

To compute precision, we define set $S_j$ as the set of items rated by the user $j$, which are not included in the training set. For user $j$, the sets $relevant_j$ (groundtruth) and $retrieved_j$ are defined as the top $p\%$ of the items in $S_j$ with the highest ratings and the highest predicted ratings, respectively. The average precision is computed as:
\begin{equation}
\text{precision} = \frac{1}{m}\sum_{j=1}^m \frac{|relevant_j| \cap |retrieved_j|}{|retrieved_j|}.
\end{equation}

We report Normalized Discounted Cumulative Gain (NDCG) and Hit Ratio (HR) to evaluate the implicit feedback prediction performance. Note that for each user, we have only one relevant item in the test set, which is the last interacted item. For each user $j$, we call this relevant item $relevant_j$. We follow the same protocols as \citep{Dong19, He17} and truncate the rank list at $10$ for both metrics. We call this list $retrieved_j$. HR measures whether the actual test item exists in the top-ranked list:
\begin{equation}
	\text{HR} = \frac{1}{m}\sum_{j=1}^m \ind(relevant_j \in retrieved_j),
\end{equation}
where $\ind(arg)=1$ if $arg$ is true, and $\ind(arg)=0$ otherwise.

NDCG is a measure of ranking quality and gives higher scores to the hits at top positions in the ranked list. Let us define $index_j$ as position of the $relevant_j$ in the ranked list $retrieved_j$. Then, NDCG is defined as:
\begin{equation}
	\text{NDCG} = \frac{1}{m}\sum_{j=1}^m \ind(relevant_j \in retrieved_j) \frac{1}{\log2(index_j+1)}.
\end{equation}

\subsubsection{Learning the input improves the prediction performance.}  In Table~\ref{t:ml100L} of section~\ref{s:exp}, we compare RMSE of InP-NCF with NCF on ml100k, ml1m, and Ichiba datasets. Here, in Table~\ref{t:inplearn}, we have an extension of those experiments, where we report the mean and standard deviation of precision and RMSE on neural networks with different number of layers. The top, middle, and bottom sections of the Table~\ref{t:inplearn}  contains results of the ml100k, ml1m, and Ichiba datasets, respectively. The notation $\bg^{u},\bg^{i},f =[a,b,c]$ means that the MLPs $\bg^{u}()$, $\bg^{i}()$, and $f()$ have $a$, $b$, and $c$ hidden layers. In other words, the neural networks become deeper from left column to the right column of Table~\ref{t:inplearn}.

By comparing NCF (R) with InP-NCF (R) and NCF (I) with InP-NCF (I), we can see that learning the input consistently improves the results and gives lower RMSE and higher precision.
Additionally, in almost all cases, our approach InP-NCF with a shallow network achieves better results than NCF (fixed input) with a deep network.

\begin{table*}
	\caption{Extension of the experiments reported in Table~\ref{t:ml100L} of section~\ref{s:exp}. We compare our approach InP-NCF with NCF (fixed inputs) in different datasets. Our method consistently achieves better results.}
	\label{t:inplearn}
	\centering
	\begin{tabular}{cccc}
	\multicolumn{4}{c}{\dotfill ml100k  \dotfill} \\
	\toprule
	\# layers & $\bg^{u},\bg^{i},f =[1,1,2]$ &$\bg^{u},\bg^{i},f =[2,2,3]$ & $\bg^{u},\bg^{i},f =[3,3,4]$ \\
	& \multicolumn{3}{c}{\dotfill RMSE  \dotfill} \\
	\midrule
	 NCF (R) & $0.903  \pm  0.003$ & $0.900  \pm  0.004$ & $0.897  \pm  0.002$ \\
	 \midrule
	 InP-NCF (R) & $\mathbf{0.892  \pm  0.003}$ & $\mathbf{0.892  \pm  0.002}$ & $\mathbf{0.894 \pm  0.001}$ \\
	 \midrule
	 NCF (I) & $0.905 \pm 0.004$ & $0.900 \pm 0.003$ & $0.896 \pm 0.002$ \\ 
	 \midrule
	 InP-NCF (I)& $\mathbf{0.892 \pm 0.004}$ & $\mathbf{0.894 \pm 0.002}$ & $\mathbf{0.895 \pm 0.002}$ \\
	 \bottomrule
	 \\
	 & \multicolumn{3}{c}{\dotfill precision  \dotfill} \\
	 \midrule
	NCF (R) &  $69.8\% \pm 0.7$ & $69.7\% \pm 0.6$ & $69.3\% \pm 0.6$ \\
	\midrule
		InP-NCF (R) & $\mathbf{69.8\% \pm 0.5}$ & $\mathbf{70.0\% \pm 0.6}$ & $\mathbf{70.0\% \pm 0.6}$\\
	\midrule
	NCF (I) & $69.4\% \pm 0.4$ & $69.5\% \pm 0.7$ & $69.7\% \pm 0.5$\\
	\midrule
	InP-NCF (I) & $\mathbf{70.0\% \pm 0.3}$ & $\mathbf{70.2\% \pm 0.3}$ & $\mathbf{69.8\% \pm 0.5}$ \\
	\bottomrule
	\vspace{2ex}
	\end{tabular}
	\begin{tabular}{cccc}
	\multicolumn{4}{c}{\dotfill ml1m  \dotfill} \\
	\toprule
	\# layers & $\bg^{u},\bg^{i},f=[1,1,2]$ &$\bg^{u},\bg^{i},f =[2,2,3]$ & $\bg^{u},\bg^{i},f =[3,3,4]$\\
	& \multicolumn{3}{c}{\dotfill RMSE  \dotfill} \\
	\midrule
	 NCF (R) & $0.858  \pm  0.000$ & $0.857  \pm  0.002$ & $0.856  \pm  0.000$\\
	 	\midrule
	 InP-NCF (R) & $\mathbf{0.845  \pm  0.001}$ & $\mathbf{0.843  \pm  0.000}$ & $\mathbf{0.845 \pm  0.002}$\\
	 \midrule
	 NCF (I) & $0.860  \pm  0.000$ & $0.859  \pm  0.001$ & $0.854  \pm  0.002$\\
	 \midrule
	 InP-NCF (I)& $\mathbf{0.846 \pm 0.000}$ & $\mathbf{0.846 \pm 0.000}$ & $\mathbf{0.844 \pm 0.001}$\\
	 \bottomrule
	 \\
	 & \multicolumn{3}{c}{\dotfill precision  \dotfill} \\
	 \midrule
	NCF (R) &  $69.4\% \pm 0.3$ & $69.7\% \pm 0.2$ & $70.0\% \pm 0.3$ \\
	\midrule
	InP-NCF (R) & $\mathbf{70.3\% \pm 0.09}$ & $\mathbf{70.3\% \pm 0.21}$ & $\mathbf{70.4\% \pm 0.16}$\\
	\midrule
	NCF (I) & $69.3\% \pm 0.3$ & $69.7\% \pm 0.1$ & $69.8\% \pm 0.3$ \\
	\midrule
	InP-NCF (I) & $\mathbf{69.9\% \pm 0.19}$ & $\mathbf{70.0\% \pm 0.15}$ & $\mathbf{70.2\% \pm 0.2}$\\
	\bottomrule
	\end{tabular}
\end{table*}

\begin{table}
\caption{Same as Table~\ref{t:inplearn}, but for Ichiba dataset.}
\label{t:inplearn1}
\centering
	\begin{tabular}{cccc}
		\multicolumn{4}{c}{\dotfill Ichiba  \dotfill} \\
	\toprule
	\# layers & $\bg^{u},\bg^{i},f =[1,1,2]$ &$\bg^{u},\bg^{i},f =[2,2,3]$ & $\bg^{u},\bg^{i},f =[3,3,4]$\\
	& \multicolumn{3}{c}{\dotfill RMSE  \dotfill} \\
	\midrule
	 NCF (R) & $0.883  \pm  0.003$ & $0.885  \pm  0.001$ & $0.884  \pm  0.004$\\
	 \midrule
	 	InP-NCF (R) & $\mathbf{0.871  \pm  0.004}$ & $\mathbf{0.874  \pm  0.003}$ & $\mathbf{0.873 \pm  0.002}$\\
	 \midrule
	 NCF (I) & $0.875  \pm  0.003$ & $0.875  \pm  0.001$ & $0.874  \pm  0.003$\\
	 \midrule
	 InP-NCF(I)& $\mathbf{0.867 \pm 0.003}$ & $\mathbf{0.867 \pm 0.004}$ & $\mathbf{0.863 \pm 0.003}$\\
	  \bottomrule
	\\
	 & \multicolumn{3}{c}{\dotfill precision  \dotfill} \\
	 \midrule
	NCF (R) &  $79.5\% \pm 0.3$ & $79.2\% \pm 0.0$ & $79.4\% \pm 0.3$\\
		\midrule
	InP-NCF (R) & $\mathbf{79.8\% \pm 0.3}$ & $\mathbf{79.6\% \pm 0.3}$ & $\mathbf{79.6\% \pm 0.03}$\\
	\midrule
	NCF (I) & $79.5\% \pm 0.2$ & $79.6\% \pm 0.2$ & $79.6\% \pm 0.2$\\
	\midrule
	InP-NCF (I) & $\mathbf{79.7\% \pm 0.2}$ & $\mathbf{80.0\% \pm 0.5}$ & $\mathbf{80.0\% \pm 0.3}$\\
	\bottomrule
	\end{tabular}
\end{table}

\begin{table*}[t]
\caption{Implementation details: optimization method (optimizer), learning rate (lr), regularization (regu.), activation function (activ.).}
\label{t:imp}
\begin{center}
\begin{tabular}[c]{ccccc@{\hspace{4ex}}cccc} 
\toprule
& \multicolumn{4}{c}{\dotfill ml100k \dotfill}  & \multicolumn{4}{c}{\dotfill ml1m \dotfill} \\
  & optimizer & lr  & regu.\ & activ.\  & optimizer & lr  & regu.\ & activ.\ \\
\midrule
InP-NCF  & RMSprop & $0.001$ & $0$ & selu & SGD & $0.0005$ & $0$ & selu \\
\midrule
DeepCF  & RMSprop & $0.001$ & $0$ & relu & SGD & $0.0005$ & $0$ & relu \\
\midrule
aSDAE &  SGD & $0.1$ & $10^{-5}$ & tanh &  RMSprop & $0.001$ & $10^{-5}$ & tanh  \\
\midrule
DHA &  SGD & $0.1$ & $10^{-5}$ & tanh   &  RMSprop & $0.001$ & $10^{-5}$ & tanh  \\
\midrule
Autorec & Adam & $0.001$ & $0.005$ & tanh & Adam & $0.001$ & $0.001$ & tanh  \\
\midrule
NeuMF & RMSprop & $0.001$ & $0.01$ & selu & RMSprop & $0.001$ & $0.01$ & selu \\
\midrule
DSSM & RMSprop & $0.001$ & $0.001$ & selu & SGD & $0.0005$ & $0$ & selu  \\
\midrule
MF & SGD & $0.1$ & $0.0001$ & NA & SGD & $0.1$ & $0.0001$ & NA  \\
\bottomrule
\end{tabular}
\begin{tabular}[c]{ccccc@{\hspace{4ex}}cccc} 
\toprule
& \multicolumn{4}{c}{\dotfill Amazon \dotfill}  & \multicolumn{4}{c}{\dotfill Ichiba \dotfill} \\
  & optimizer & lr  & regu.\ & activ.\  & optimizer & lr  & regu.\ & activ.\ \\
\midrule
InP-NCF & RMSprop & $0.001$ & $0$ & selu & SGD & $0.0005$ & $0$ & selu \\
\midrule
DeepCF & RMSprop & $0.001$ & $0$ & relu & SGD & $0.0005$ & $0$ & relu \\
\midrule
aSDAE &  OM & OM & OM & OM &  OM & OM & OM & OM  \\
\midrule
DHA &  OM & OM & OM & OM &  OM & OM & OM & OM  \\
\midrule
Autorec & RMSprop & $0.001$ & $0$ & selu & Adam & $0.001$ & $0.001$ & tanh  \\
\midrule
NeuMF & RMSprop & $0.00001$ & $0.001$ & selu & SGD & $0.0001$ & $0.01$ & selu \\
\midrule
DSSM&  NA & NA & NA & NA & SGD & $0.0005$ & $0.001$ & selu\\
\midrule
MF & SGD & $0.1$ & $0.0001$ & NA & SGD & $0.1$ & $0.0001$ & NA  \\
\bottomrule
\end{tabular}
\end{center}
\vspace{-1ex}
\end{table*}

\subsection{Experimental settings}
We implement our method using Keras with TensorFlow 2 backend. We ran all the experiments on a $12$GB GPU. For each method, we tried a set of activation functions (relu, selu, and tanh), a range of learning rates and regularization parameters from $10^{-1}$ to $10^{-5}$, a set of optimizers (Adam, SGD, and RMSprop), and picked the one that works best. For a fair comparison, all autoencoder methods have the same structure (\# of layers, neurons, etc.). Table~\ref{t:imp} lists these details for all the methods.

In implicit feedback prediction, we use the same hyper-parameters used in DeepCF code, which is available online. In the following, we give more detailed information about the rating prediction, such as structure of the neural network and other hyper-parameters of each method. The notation $[a,b,c]$ denotes a network with three fully connected layers, where $a,b,$ and $c$ are the number of neurons in each layer.
\begin{itemize}
 \item \textbf{InP-NCF}. In ml100k, both $\bg^{u}$ and $\bg^{i}$ are MLPs with $[100,200,100]$ and $f$ is an MLP with $[500,200,100,1]$. In ml1m, both $\bg^{u}$ and $\bg^{i}$ are MLPs with $[500,300,100]$ and $f$ is an MLP with $[500,200,100,1]$. In Amazon, both $g^{u}$ and $g^{i}$ are MLPs with $[500,300,100]$ and $f$ is an MLP with $[200,200,100,1]$. In Ichiba, both $\bg^{u}$ and $\bg^{(i)}$ are MLPs with $[100,100,100]$ and $f$ is an MLP with $[100,100,100,1]$.
 
Note that we use a post-input optimization over parameters of the networks and the input. The details of the network optimization can be found in Table~\ref{t:imp}. For the post-input optimization, we use SGD with the learning rate $0.1$, $0.1$, $0.1$, and $0.5$ in ml100k, ml1m, Amazon, and Ichiba datasets, respectively.
 \item \textbf{DeepCF}. We use the same structure as InP-NCF for the MLPs in the two branches. In ml100k, ml1m, Amazon, and ichiba, the MLPs are [$500$, $200$, $100$], [$1000$, $500$, $300$, $100$], [$500$, $300$, $100$], [$200$, $100$, $100$].
 \item \textbf{DHA \cite{Li18}  and aSDAE \cite{Dong17}}. The encoders are $[500, 200, 100]$, $[1000, 500, 300, 100]$, $[500, 300, 100]$, and  $[500, 300, 100]$ in ml100k, ml1m, Amazon, and Ichiba datasets. These two methods have a large number of hyper-parameters and we tried a large range of values to get the best results. We have set $\lambda_2=\lambda_3=0.5$ and $\lambda_1=1$. We use SGD with learning rate of $0.1$ for the optimization over the variables $\bU$ and $\bV$.
\item \textbf{HIRE \cite{Liu19}}. We used the code provided by the authors without changing the hyper-parameters.
 \item \textbf{NeuMF  \cite{He17}}. This method has one deep and one shallow branches. The structure of the deep branch is the same as the encoder network of the DHA in all datasets.
 \item \textbf{Autorec  \cite{Sedhain15}} . We implemented the I-Autorec, which reconstructs the ratings of the items. This method overfits to the training data fast, even using small networks. The autoencoder is $[m,500,m]$ in ml100k, as suggested by the original paper, and $[m,500,300,100,300,500,m]$ in ml1m, Amazon, and Ichiba datasets, where $m$ is the number of users.
\item \textbf{DSSM  \cite{Huang13} } Each branch has the same structure as the encoder network of DHA.
  
\end{itemize}

\end{document}

%% file: abbr.tex
\usepackage{bbm}
\newcommand{\ind}{\ensuremath{\mathbbm{1}}}

\newcommand{\bm}{\ensuremath{\mathbf{m}}}

\newcommand{\bR}{\ensuremath{\mathbf{R}}}
\newcommand{\bbf}{\ensuremath{\mathbf{f}}}
\newcommand{\bh}{\ensuremath{\mathbf{h}}}
\newcommand{\bg}{\ensuremath{\mathbf{g}}}
\newcommand{\bW}{\ensuremath{\mathbf{W}}}

\newcommand{\bU}{\ensuremath{\mathbf{U}}}
\newcommand{\bV}{\ensuremath{\mathbf{V}}}
\newcommand{\bz}{\ensuremath{\mathbf{z}}}
\newcommand{\bH}{\ensuremath{\mathbf{H}}}

\newcommand{\bE}{\ensuremath{\mathbf{E}}}
\newcommand{\bI}{\ensuremath{\mathbf{I}}}
\newcommand{\bA}{\ensuremath{\mathbf{A}}}

\newcommand{\bv}{\ensuremath{\mathbf{v}}}

\newcommand{\by}{\ensuremath{\mathbf{y}}}

\newcommand{\bZ}{\ensuremath{\mathbf{Z}}}
\newcommand{\bP}{\ensuremath{\mathbf{P}}}

\newcommand{\btheta}{\ensuremath{\boldsymbol{\theta}}}

\newcommand{\bbR}{\ensuremath{\mathbb{R}}}